\documentclass[10pt,twocolumn,twoside,journal]{IEEEtran}

\usepackage[noadjust]{cite}

\ifCLASSINFOpdf
  \usepackage[pdftex]{graphicx}
  \DeclareGraphicsExtensions{.pdf,.jpeg,.png,.tif,.tiff}
\else

\fi

\usepackage[cmex10]{amsmath}

\ifCLASSOPTIONcompsoc
    \usepackage[tight,normalsize,sf,SF]{subfigure}
\else
    \usepackage[tight,footnotesize]{subfigure}
\fi

\usepackage[latin1]{inputenc}
\usepackage{amsthm, amsfonts, amssymb, mathrsfs, amsmath}
\usepackage{amsmath}
  {
      \theoremstyle{plain}
      
  }
\usepackage[pdftex]{graphicx}
\usepackage{nicefrac}
\usepackage{mathtools}
\usepackage{multirow}
\usepackage[ruled]{algorithm2e}
\usepackage[mathcal]{euscript}

\usepackage[draft]{fixme}
\usepackage{enumerate}
\usepackage{dsfont}

\usepackage{tikz}
\allowdisplaybreaks[1]

\newcommand{\bM}{{\boldsymbol{M}}}
\newcommand{\bL}{{\boldsymbol{L}}}
\newcommand{\bS}{{\boldsymbol{S}}}
\newcommand{\bZ}{{\boldsymbol{Z}}}
\newcommand{\bB}{{\boldsymbol{B}}}
\newcommand{\bT}{{\boldsymbol{\Theta}}}

\newcommand{\bU}{{\boldsymbol{U}}}
\newcommand{\bV}{{\boldsymbol{V}}}
\newcommand{\bSi}{{\boldsymbol{\Sigma}}}

\newcommand{\bX}{{\boldsymbol{X}}}
\newcommand{\bx}{{\boldsymbol{x}}}
\newcommand{\bu}{{\boldsymbol{u}}}
\newcommand{\bv}{{\boldsymbol{v}}}
\newcommand{\bz}{{\boldsymbol{z}}}
\newcommand{\by}{\boldsymbol{y}}

\newcommand{\bdel}{\boldsymbol{\delta}}
\newcommand{\be}{\boldsymbol{e}}

\newcommand{\bhx}{{\boldsymbol{\widehat{x}}}}

\newcommand{\g}{\textsl{g}}

\newcommand{\bg}{\textbf{\textsl{g}}}

\makeatletter
\renewcommand{\@algocf@capt@plain}{above}
\makeatother

\DeclareMathOperator*{\argmin}{arg\,min}


\usepackage{amsthm}
\newtheorem{theorem}{Theorem}[section]

\newtheorem{proposition}[theorem]{Proposition}

\begin{document}
\title{Incorporating Prior Information in Compressive Online Robust Principal Component Analysis}

\author{Huynh Van Luong,~Nikos Deligiannis,~J\"{u}rgen Seiler,~S{\o}ren Forchhammer,~and~Andr\'{e} Kaup

\thanks{H. V. Luong, J. Seiler, and A. Kaup are with the Chair of Multimedia Communications and Signal Processing, Friedrich-Alexander-Universit\"{a}t Erlangen-N\"{u}rnberg, 91058 Erlangen, Germany (e-mail: huynh.luong@fau.de, juergen.seiler@fau.de, and andre.kaup@fau.de).}
\thanks{N. Deligiannis is with the Department of Electronics and Informatics, Vrije Universiteit Brussel, 1050 Brussels, and also with iMinds, 9050 Ghent, Belgium (e-mail: ndeligia@etro.vub.ac.be).}
\thanks{S. Forchhammer is with the Department of Photonics Engineering, Technical University of Denmark, 2800 Lyngby, Denmark (e-mail: sofo@fotonik.dtu.dk).}
}

%
\maketitle

\begin{abstract}
	We consider an online version of the robust Principle Component Analysis (PCA), which arises naturally in time-varying source separations such as video foreground-background separation. This paper proposes a compressive online robust PCA with prior information for recursively separating a sequences of frames into sparse and low-rank components from a small set of measurements. In contrast to conventional batch-based PCA, which processes all the frames directly, the proposed method processes measurements taken from each frame. Moreover, this method can efficiently incorporate multiple prior information, namely previous reconstructed frames, to improve the separation and thereafter, update the prior information for the next frame. We utilize multiple prior information by solving $n\text{-}\ell_{1}$ minimization for incorporating the previous sparse components and using incremental singular value decomposition ($\mathrm{SVD}$) for exploiting the previous low-rank components. We also establish theoretical bounds on the number of measurements required to guarantee successful separation under assumptions of static or slowly-changing low-rank components. Using numerical experiments, we evaluate our bounds and the performance of the proposed algorithm. 
	In addition, we apply the proposed algorithm to online video foreground and background separation from compressive measurements. Experimental results show that the proposed method outperforms the existing methods. 
\end{abstract}
\begin{keywords}
	Prior information, compressed sensing, sparse signal, $n$-$\ell_{1}$ minimization, compressive measurements.
\end{keywords}


\section{Introduction}\label{sec:intro}
\IEEEPARstart{P}{rincipal} Component Analysis (PCA) is widely used in statistical data analysis in a number of applications. A modified version of Principal Component Analysis (PCA), namely, robust PCA (RPCA) \cite{CandesRPCA,JWright09,ZHLin09}, is a well-known tool to efficiently capture properties of signals of interest. A potential application of RPCA is video foreground and background separation, e.g., in video surveillance and online tracking brain activation from MRI sequences. The goal of RPCA is to decompose the high-dimensional matrix $\bM$ into the sum of unknown sparse $\bS$ and low-rank $\bL$ components. Principal Component Pursuit (PCP) \cite{CandesRPCA} recovers $\bL$ and $\bS$ by solving the following program:
\begin{equation}\label{PCP}
\min_{\bL,\bS} \|\bL\|_{*}+\lambda\|\bS\|_{1} \mathrm{~subject~to~}\bL+\bS=\bM,
\end{equation}
where $\|\cdot\|_{*}$ is the matrix nuclear norm (sum of singular values). 
For example, in video separation, a video sequence is separated into slowly-changing background as the low-rank term $\bL$ and foreground as the sparse term $\bS$, which is the quantity of interest in this paper. However, RPCA methods \cite{CandesRPCA,JWright09,ZHLin09} need to access all data frames, e.g., all frames of a video, which is batch-based and not practical. 

\textbf{Problem}. We consider a version of RPCA for the recursive separation problem of a time-varying sequence of frames that processes one frame per time instance and works on a small set of measurements to separate large-scale data. In more detail, we introduce a time instance $t$, at which we have observed $\bM_{t}\hspace{-0.2pt}=\hspace{-0.2pt}\bL_{t}\hspace{-0.2pt}+\hspace{-0.2pt}\bS_{t}$, where $\bS_{t}\hspace{-0.2pt}=\hspace{-0.2pt}[\bx_{1}~ \bx_{1}~ ...~ \bx_{t}]$ and $\bL_{t}\hspace{-0.2pt}=\hspace{-0.2pt}[\bv_{1}~ \bv_{2}~ ...~ \bv_{t}]$, where $[\cdot]$ denotes a matrix and $\bx_{t}, \bv_{t}\in \mathbb{R}^{n}$ are a column component in $\bS_{t}$ and $\bL_{t}$, respectively. In this work, we consider compressive measurements of the observed $\bx_{t}+\bv_{t}$ at a time $t$, i.e., we observe $\by_{t}\hspace{-0.2pt}=\hspace{-0.2pt}\mathbf{\Phi}(\bx_{t}\hspace{-0.2pt}+\hspace{-0.2pt}\bv_{t})$, where $\mathbf{\Phi}\hspace{-0.2pt}\in\hspace{-0.2pt} \mathbb{R}^{m\times n} $~$(m\hspace{-0.2pt}<\hspace{-0.2pt}n)$ is a random projection~\cite{CandesTIT06,DonohoCOM06}, and that $\bL_{t-1}$ and $\bS_{t-1}$ have been given for time instance $t\hspace{-0.2pt}-\hspace{-0.2pt}1$. Following Problem \eqref{PCP}, we want to solve a problem at time $t$ as follows
\begin{equation}\label{onlinePCP}
\min_{\bx_{t},\bv_{t}} \Big\|[\bL_{t-1}~\bv_{t}]\Big\|_{*}\hspace{-0.2pt}+\hspace{-0.1pt}\lambda\|\bx_{t}\|_{1} \mathrm{~subject~to~}\by_{t}\hspace{-0.2pt}=\hspace{-0.2pt}\mathbf{\Phi}(\bx_{t}+\bv_{t}),
\end{equation}
where $\bL_{t-1}\hspace{-0.4pt}=\hspace{-0.2pt}[\bv_{1}~ \bv_{2}~ ...~ \bv_{t-1}]$ and $\bS_{t-1}\hspace{-0.4pt}=\hspace{-0.2pt}[\bx_{1}~ \bx_{1}~ ...~ \bx_{t-1}]$ are given. We aim at solving Problem \eqref{onlinePCP} given prior information $\bL_{t-1}$, $\bS_{t-1}$, and $\mathbf{\Phi}$ as formulated in Sec. \ref{formulateCORPCA}.

\textbf{Related Works}. Some previous approaches \cite{LauraGrouse10,JHe12,Feng13,Mansour15}
were proposed for online estimating low-dimensional subspaces from randomly subsampled data. However, they have not taken the sparse component $\bx_{t}$ into account. In addition, compressive ReProCS in \cite{GuoQV14,QiuVLH14} was proposed to recover the sparse component for foreground extraction with compressive measurements. Compressive ReProCS \cite{GuoQV14} considers recovering $\bx_{t}$ from $\by_{t}\hspace{-0.2pt}=\hspace{-0.2pt}\mathbf{\Phi}_{t}(\boldsymbol{A}\bx_{t}\hspace{-0.1pt}+\hspace{-0.1pt}\boldsymbol{B}\bv_{t})$ \cite{GuoQV14}. It can be noted here that $\mathbf{\Phi}_{t}\hspace{-0.2pt}\in \hspace{-0.2pt}\mathbb{R}^{n\times n}$ and $\boldsymbol{A}\hspace{-0.2pt}=\hspace{-0.2pt}\boldsymbol{B }\hspace{-0.2pt}\in \hspace{-0.2pt}\mathbb{R}^{m\times n}$. However, the background $\bv_{t}$ is not recovered after each instance and no explicit condition on how many measurements are required for successful recovery is considered. In this work, we focus on recovering both foreground and background components from compressive measurements. 

Recently, some related approaches to foreground extraction in \cite{MotaTSP16}, \cite{WarnellTIP15} utilizing an adaptive rate of the compressive measurement were introduced, yet they considered the background static. Another work in \cite{Rodriguez16} proposed an incremental PCP method that processes one frame at a time, however it did use information of the full frame, rather than compressive measurements. Compressive PCP \cite{JWright13} considered the separation on compressive measurements as a batch-based method.

Furthermore, the problem of reconstructing a time sequence of sparse signals with prior information is also playing an important role in the context of online RPCA~\cite{QiuV11,GuoQV14,QiuVLH14}. There were several studies on sparse signal recovery with prior information from low-dimensional measurements~\cite{MotaTSP16,QiuV11,GuoQV14,QiuVLH14,VaswaniZ16}. The study in~\cite{VaswaniZ16} provided a comprehensive overview of the domain, reviewing a class of recursive algorithms for recovering a time sequence of sparse signals from a small number of measurements. 
The studies in~\cite{QiuV11,GuoQV14,QiuVLH14} used modified-CS~\cite{NVaswaniTSP10} to leverage prior knowledge under the condition of slowly varying support and signal values. However, all the above can not exploit multiple prior information gleaned from multiple previously recovered frames.

\textbf{Contributions}. We propose a compressive online RPCA with multiple prior information named, CORPCA, which exploits the information for previously recovered sparse components by utilizing RAMSIA \cite{LuongICIP16} and leverages the slowly-changing characteristics of low-rank components via an incremental $\mathrm{SVD}$ \cite{Brand02}. We perform the separation in an online manner by minimizing, (i) $n$-$\ell_{1}$ function \cite{LuongICIP16} for the sparse part; and (ii) the rank of a matrix of the low-rank part. Thereafter, the new recovered sparse and low-rank components are updated as a new prior knowledge for the next processing instance.    
Furthermore, we derive the number of measurements that guarantee successful compressive separation. As such, we set a bound on the number of required measurements based on the current foreground-background and prior information. 

The rest of this paper is organized as follows. Our problem formulation and the proposed algorithm  are presented in Sec. \ref{formulateCORPCA}. We describe performance guarantees of the proposed method including the measurement bound and the proof sketch in Sec. \ref{performGuarantees}. We demonstrate numerical results of the proposed method in Sec. \ref{Experiment}.

\section{Problem Formulation and Algorithm}\label{formulateCORPCA}
In this section, we formulate explicitly our problem (in Sec. \ref{problem}), following by our proposed algorithm in Sec. \ref{corpca}. 
\subsection{Problem Formulation}\label{problem}
The proposed CORPCA algorithm is based on RAMSIA \cite{LuongICIP16}, our previously proposed agorithm that uses $n$-$\ell_1$  minimization with adaptive weights to recover a sparse signal using multiple side information signals.
The objective function of RAMSIA \cite{LuongICIP16} is given by:
\begin{equation}\label{n-l1minimizationGlobal}
\min_{\bx}\hspace{-0.2pt}\Big\{\hspace{-0.2pt}H(\bx)\hspace{-0.2pt}=\hspace{-0.2pt}\frac{1}{2}\|\mathbf{\Phi}\bx-\by\|^{2}_{2} + \lambda \hspace{-0.2pt}\sum\limits_{j=0}^{J}\hspace{-0.2pt}\beta_{j}\|\mathbf{W}_{j}(\bx-\bz_{j})\|_{1}\Big\},
\end{equation}
where $\lambda>0$, $\beta_{j}\hspace{-0.2pt}>\hspace{-0.2pt}0$ are weights among side information signals, and $\mathbf{W}_{j}$ is a diagonal weighting matrix for each side information $\bz_{j}$, $\mathbf{W}_{j}\hspace{-0.2pt}=\hspace{-0.2pt}\mathrm{diag}(w_{j1},w_{j2},...,w_{jn})$, wherein $w_{ji}\hspace{-0.2pt}>\hspace{-0.2pt}0$ is the weight in $\mathbf{W}_{j}$ at index $i$ for a given $\bz_{j}$ with $\bz_{0}\hspace{-0.2pt}=\hspace{-0.2pt}\mathbf{0}$.

The proposed CORPCA method aims at processing one frame per time instance and incorporating prior information for both sparse components and low-rank components. At time instance $t$, we observe $\by_{t}=\mathbf{\Phi}(\bx_{t}+\bv_{t})$ with $\by_{t}\in\mathbb{R}^{ m}$.
Let $\bZ_{t-1}:=\{\bz_{1},...,\bz_{J}\}$, which is a set of $\bz_{j}\in \mathbb{R}^{n}$, and $\bB_{t-1}\in \mathbb{R}^{n\times d}$ denote prior information for $\bx_{t}$ and $\bv_{t}$, respectively. We can form the prior information $\bZ_{t-1}$ and $\bB_{t-1}$ based on the reconstructed $\{ \bx_{1}, ..., \bx_{t-1}\} $ and $\{\bv_{1}, ..., \bv_{t-1}\}$, which is discussed in Sec. \ref{corpca}.

To solve the problem in \eqref{onlinePCP}, we fomulate the objective function of CORPCA as
\begin{align}\label{CORPCAminimization}
\hspace{-0.2pt}\min_{\bx_{t},\bv_{t}}\hspace{-0.2pt}&\Big\{\hspace{-0.2pt}H(\bx_{t},\bv_{t}|\by_{t},\bZ_{t-1},\bB_{t-1})\hspace{-0.2pt}=\hspace{-0.2pt}\frac{1}{2}\|\mathbf{\Phi}(\bx_{t}+\bv_{t})-\by_{t}\|^{2}_{2} \nonumber \\
&+\lambda \mu\hspace{-0.2pt}\sum\limits_{j=0}^{J}\hspace{-0.2pt}\beta_{j}\|\mathbf{W}_{j}(\bx_{t}-\bz_{j})\|_{1}+\mu\Big\|[\bB_{t-1}~ \bv_{t}]\Big\|_{*}\Big\},
\end{align}
where $\mu>0$. It can be seen that when $\bv_{t}$ is static, Problem \eqref{CORPCAminimization} would become Problem \eqref{n-l1minimizationGlobal}. When $\bx_{t}$ and $\bv_{t}$ are batch variables 
without taking the prior information, $\bZ_{t-1}$ and $\bB_{t-1}$, and the projection  $\mathbf{\Phi}$ into account, Problem \eqref{CORPCAminimization} becomes Problem \eqref{PCP}.
\subsection{CORPCA Algorithm} 
\label{corpca}
We now describe how to solve Problem \eqref{CORPCAminimization} and thereafter update the prior information, $\bZ_{t-1}$ and $\bB_{t-1}$, for the next instance. First, we solve Problem \eqref{CORPCAminimization} given $\bZ_{t-1}$ and $\bB_{t-1}$. Let us denote $f(\bv_{t},\bx_{t})=(1/2)\|\mathbf{\Phi}(\bx_{t}+\bv_{t})-\by_{t} \|^{2}_{2}$, $g(\bx_{t})=\lambda \hspace{-0.2pt}\sum_{j=0}^{J}\hspace{-0.2pt}\beta_{j}\|\mathbf{W}_{j}(\bx_{t}-\bz_{j})\|_{1}$, and $\textsl{h}(\bv_{t})=\|[\bB_{t-1}~ \bv_{t}]\|_{*}$. The solution of \eqref{CORPCAminimization} using proximal gradient methods \cite{Beck09,JWright09} gives that $\bx_{t}^{(k+1)}$ and $\bv_{t}^{(k+1)}$ at iteration $k+1$ can be iteratively computed via the soft thresholding operator \cite{Beck09} for $\bx_{t}$ and the single value thresholding operator \cite{Cai10} for $\bv_{t}$:
\begin{equation}
\label{vtProximal}
\hspace{-0.0pt}\bv_{t}^{(\hspace{-1pt}k+1\hspace{-1pt})}\hspace{-4pt}=\hspace{-2pt}\argmin_{\bv_{t}}\hspace{-0.0pt}\Big\{\hspace{-0.2pt}\mu\textsl{h}(\bv_{t})\hspace{-0.1pt}+\hspace{-0.2pt}\Big\|\bv_{t}\hspace{-2pt}-\hspace{-0.2pt}\Big(\bv_{t}^{(\hspace{-0.1pt}k\hspace{-0.1pt})}\hspace{-0.4pt}-\hspace{-2pt}\frac{1}{2}\nabla_{\hspace{-0.3pt}\bv_{t}\hspace{-0.1pt}}f(\hspace{-0.2pt}\bv_{t}^{(\hspace{-0.1pt}k\hspace{-0.1pt})}\hspace{-2pt},\bx_{t}^{(\hspace{-0.1pt}k\hspace{-1pt})}\hspace{-2pt})\hspace{-2pt}\Big)\Big\|_{2}^{2}\hspace{-1pt}\hspace{-0.0pt}\Big\},\\\hspace{-0.7pt}
\end{equation}
\begin{equation}
\label{xtProximal}
\hspace{-0.0pt}\bx_{t}^{(\hspace{-1pt}k\hspace{-1pt}+1)}\hspace{-0.4pt}=\hspace{-2pt}\argmin_{\bx_{t}}\hspace{-0.0pt}\Big\{\hspace{-2pt}\mu g(\bx_{t})\hspace{-0.1pt}+\hspace{-2pt}\Big\|\bx_{t}\hspace{-2pt}-\hspace{-2pt}\Big(\bx_{t}^{(\hspace{-1pt}k\hspace{-1pt})}\hspace{-0.4pt}-\hspace{-0.2pt}\frac{1}{2}\nabla_{\hspace{-0.3pt}\bx_{t}} \hspace{-1pt}f(\hspace{-2pt}\bv_{t}^{(\hspace{-1pt}k\hspace{-1pt})}\hspace{-2pt},\bx_{t}^{(\hspace{-1pt}k\hspace{-0.1pt})}\hspace{-2pt})\hspace{-2pt}\Big)\Big\|_{2}^{2}\hspace{-1pt}\hspace{-0.0pt}\Big\}.\hspace{-0.8pt}
\end{equation}
\setlength{\textfloatsep}{0pt}
\begin{algorithm}[t!]\label{CORPCAAlg}
	\DontPrintSemicolon \SetAlgoLined
	\textbf{Input}: $\by_{t},~\mathbf{\Phi},~\bZ_{t-1},~\bB_{t-1}$;\\
	\textbf{Output}: $\bx_{t},~\bv_{t},~\bZ_{t},~\bB_{t}$;\\
	\tcp{Initialization.}
	$\bx_{t}^{(-1)}\hspace{-0.2pt}=\hspace{-0.2pt}\bx_{t}^{(0)}\hspace{-0.2pt}=\hspace{-0.2pt}\mathbf{0}$;	$\bv_{t}^{(-1)}\hspace{-0.2pt}=\hspace{-0.2pt}\bv_{t}^{(0)}\hspace{-0.2pt}=\hspace{-0.2pt}\mathbf{0}$;
	$\xi_{-1}\hspace{-0.2pt}=\xi_{0}\hspace{-0.2pt}=\hspace{-0.2pt}1$; $\mu_{0}\hspace{-0.2pt}=\hspace{-0.2pt}0$; $\bar{\mu}\hspace{-0.2pt}>\hspace{-0.2pt}0$; $\lambda>0$; $0\hspace{-0.2pt}<\hspace{-0.2pt}\epsilon\hspace{-0.2pt}<\hspace{-0.2pt}1$; $k\hspace{-0.2pt}=\hspace{-0.2pt}0$; $g_{1}(\cdot)\hspace{-0.2pt}=\hspace{-0.2pt}\|\cdot\|_{1}$; \\
	\While{not converged}{
		\tcp{Solving the problem.}
		$\widetilde{\bv_{t}}^{(k)}\hspace{-0.2pt}=\bv_{t}^{(k)}\hspace{-0.2pt}+\hspace{-0.2pt}\frac{\xi_{k-1}-1}{\xi_{k}}(\bv_{t}^{(k)}\hspace{-0.2pt}-\hspace{-0.2pt}\bv_{t}^{(k-1)})$; \\
		$\widetilde{\bx_{t}}^{(k)}\hspace{-0.2pt}=\bx_{t}^{(k)}\hspace{-0.2pt}+\hspace{-0.2pt}\frac{\xi_{k-1}-1}{\xi_{k}}(\bx_{t}^{(k)}\hspace{-0.2pt}-\hspace{-0.2pt}\bx_{t}^{(k-1)})$; \\
		$\nabla_{\bv_{t}} f(\widetilde{\bv_{t}}^{(k)},\widetilde{\bx_{t}}^{(k)})=\nabla_{\bx_{t}}f(\widetilde{\bv_{t}}^{(k)},\widetilde{\bx_{t}}^{(k)})=\mathbf{\Phi}^{\mathrm{T}}\Big(\mathbf{\Phi} (\widetilde{\bv_{t}}^{(k)}+\widetilde{\bx_{t}}^{(k)})-\by_{t}\Big)$; \\	
		
		%
		%

		\hspace{-0.0pt}$(\bU_{t},\bSi_{t},\bV_{t})\hspace{-0.1pt}=\hspace{-0.1pt}\hspace{-0.2pt}\mathrm{incSVD}\Big(\hspace{-0.0pt}\Big[\bB_{t-1}~\Big(\widetilde{\bv_{t}}^{(k)}\hspace{-0.2pt}-\hspace{-0.2pt}\frac{1}{2}\nabla_{\bv_{t}} f(\widetilde{\bv_{t}}^{(k)},\widetilde{\bx_{t}}^{(k)})\Big)\Big]\Big)$; \\
		
		$\bT_{t}\hspace{-0.2pt}=\hspace{-0.2pt}\bU_{t}\mathbf{\Gamma}_{\frac{\mu_{k}}{2}g_{1}}(\bSi_{t})\bV_{t}^{T}$; \\
		$\bv_{t}^{(k+1)}\hspace{-0.2pt}=\bT_{t}(:,\mathrm{end})$; \\
		$\bx_{t}^{(k+1)}\hspace{-0.2pt}=\hspace{-0.2pt}\Gamma_{\frac{\mu_{k}}{2}g}\Big(\widetilde{\bx_{t}}^{(k)}-\frac{1}{2}\nabla_{\bx_{t}} f(\widetilde{\bv_{t}}^{(k)},\widetilde{\bx_{t}}^{(k)})\Big)$;  where $\Gamma_{\frac{\mu_{k}}{2}g}(\cdot)$ is given as in RAMSIA \cite{LuongICIP16};\\
	\hspace{-0.0pt}\tcp{\hspace{-0.6pt}Computing the updated weights \cite{LuongICIP16}\hspace{-0.2pt}.}
		$w_{ji} \hspace{-0.2pt}=\dfrac{n(|x_{ti}^{(k+1)}\hspace{-0.2pt}-\hspace{-0.2pt}z_{ji}|\hspace{-0.2pt}+\hspace{-0.1pt}\epsilon)^{-1}}{\sum\limits_{l=1}^{n}(|x_{tl}^{(k+1)}\hspace{-0.2pt}-\hspace{-0.2pt}z_{jl}|\hspace{-0.2pt}+\hspace{-0.1pt}\epsilon)^{-1}}$;
		\\
		$\beta_{j}\hspace{-0.2pt} =\dfrac{\Big(||\mathbf{W}_{j}(\bx_{t}^{(k+1)}\hspace{-0.2pt}-\hspace{-0.2pt}\bz_{j}^{(k+1)})||_{1}\hspace{-0.2pt}+\hspace{-0.1pt}\epsilon\Big)^{-1}}{\sum\limits_{l=0}^{J}\hspace{-0.2pt}\Big(||\mathbf{W}_{l}(\bx_{t}^{(k+1)}\hspace{-0.2pt}-\hspace{-0.2pt}\bz_{l}^{(k+1)})||_{1}\hspace{-0.2pt}+\hspace{-0.1pt}\epsilon\Big)^{-1}}$;
		
		$\xi_{k+1}=(1+\sqrt{1+4\xi_{k}^{2}})/2$; $\mu_{k+1}=\max(\epsilon\mu_{k},\bar{\mu})$\\
		$k=k+1$; \\
	}
	\tcp{Updating prior information.}
	$\bZ_{t}:=\{\bz_{j}=\bx^{(k+1)}_{t-J+j}\}_{j=1}^{J}$;\\
	$\bB_{t}=\bU_{t}(:,1:d)\mathbf{\Gamma}_{\frac{\mu_{k}}{2}g_{1}}(\bSi_{t})(1:d,1:d)\bV_{t}(:,1:d)^{\mathrm{T}}$;\\
	\tcp{Output.}	
	\Return $\bx_{t}^{(k+1)},~\bv_{t}^{(k+1)},~\bZ_{t},~\bB_{t}$;
	\caption{The proposed CORPCA algorithm.}
\end{algorithm}
	
The proposed CORPCA algorithm is described in Algorithm \ref{CORPCAAlg}, where as suggested in \cite{JWright09}, the convergence of Algorithm \ref{CORPCAAlg} is determined by $\|\partial H(\bx_{t},\bv_{t})|_{\bx_{t}^{(k+1)},\bv_{t}^{(k+1)}}\|_{2}^{2}\hspace{-0.2pt}<\hspace{-0.2pt}2*10^{-7}\|(\bx_{t}^{(k+1)},\bv_{t}^{(k+1)})\|_{2}^{2}$ and the proximal operator $\mathbf{\Gamma}_{\tau g_{1}}(\cdot)$ is defined by:
\begin{equation}\label{l1-proximalOperatorMatrix}
\mathbf{\Gamma}_{\tau g_{1}}(\bX) = \argmin_{\bV }\Big\{ \tau g_{1}(\bV) + \frac{1}{2}||\bV-\bX||^{2}_{2}\Big\}.
\end{equation}

\textbf{Prior Information Update}. Updating $\bZ_{t}$ and $\bB_{t}$ is carried out after each time instance so as to be used for subsequent processing. Because of the correlated characteristics of recursive frames, we can simply update prior information $\bZ_{t}$ by the $J$ latest previous $\bx_{t}$, i.e., $\bZ_{t}:=\{\bz_{j}=\bx_{t-J+j}\}_{j=1}^{J}$. For prior information $\bB_{t}\in \mathbb{R}^{n\times d}$, we may have an adaptive update, where we can adjust the value $d$ or keep a constant number $d$ of the columns of $\bB_{t}$. For each time instance, we use incremental singular decomposition $\mathrm{SVD}$ \cite{Brand02}, namely, $\mathrm{incSVD}(\cdot)$ in Algorithm \ref{CORPCAAlg}. When we update $\bB_{t}=\bU_{t}\mathbf{\Gamma}_{\frac{\mu_{k}}{2}g_{1}}(\bSi_{t})\bV_{t}^{\mathrm{T}}$, the dimension of $\bB_{t}$ is increased, where $\bB_{t}\in \mathbb{R}^{n\times (d+1)}$, after each instance. However, we need to stay at a reasonable constant number of $d$, thus we update $\bB_{t}=\bU_{t}(:,1:d)\mathbf{\Gamma}_{\frac{\mu_{k}}{2}g_{1}}(\bSi_{t})(1:d,1:d)\bV_{t}(:,1:d)^{\mathrm{T}}$. The advantage of using $\mathrm{incSVD}$ is that the computational cost of $\mathrm{incSVD}(\cdot)$ is lower than conventional SVD \cite{Brand02,Rodriguez16} since we only compute the full $\mathrm{SVD}$ of the middle matrix with size $(d+1)\times(d+1)$, where $d\ll n$, instead of $n\times(d+1)$. The computation of $\mathrm{incSVD}(\cdot)$ is elaborated below.  

Our goal is to compute $\mathrm{incSVD}[\bB_{t-1} ~\bv_{t}]$, i.e., $[\bB_{t-1} ~\bv_{t}]=\bU_{t}\bSi_{t}\bV_{t}^\mathrm{T}$. Given matrix $\bB_{t-1}\in \mathbb{R}^{n\times d}$, we take the $\mathrm{SVD}$ of $\bB_{t-1}$ to obtain $\bB_{t-1}=\bU_{t-1}\bSi_{t-1}\bV_{t-1}^\mathrm{T}$. Therefore, we can derive $(\bU_{t},\bSi_{t},\bV_{t})$ via $(\bU_{t-1},\bSi_{t-1},\bV_{t-1})$ and $\bv_{t}$. We recast the matrix of $[\bB_{t-1} ~\bv_{t}]$ by
\begin{equation}\label{incSVD}
[\bB_{t-1} ~\bv_{t}]\hspace{-0.2pt}=\hspace{-0.2pt}[\bU_{t-1}~\dfrac{\bdel_{t}}{\|\bdel_{t}\|_{2}}]\hspace{-0.2pt}\cdot\hspace{-0.2pt}\left[\hspace{-0.2pt}\begin{array}{cc} \bSi_{t-1}\hspace{-0.2pt} &\hspace{-0.2pt} \be_{t}\\ \textbf{0}^\mathrm{T}\hspace{-0.2pt} & \hspace{-0.2pt}\|\bdel_{t}\|_{2} \end{array}\hspace{-0.2pt}\right]
\hspace{-0.2pt}\cdot\hspace{-0.2pt}\left[\hspace{-0.2pt}\begin{array}{cc} \bV_{t-1}^\mathrm{T} \hspace{-0.2pt}&\hspace{-0.2pt} \textbf{0}\\ \textbf{0}^\mathrm{T} \hspace{-0.2pt}&\hspace{-0.2pt} 1 \end{array}\hspace{-0.2pt}\right],
\end{equation}
where $\be_{t}=\bU_{t-1}^\mathrm{T}\bv_{t}$ and $\bdel_{t} =\bv_{t}-\bU_{t-1}\be_{t}$.
By taking $\mathrm{SVD}$ of the matrix in between the right side of \eqref{incSVD}, we yield $\left[\begin{array}{cc} \bSi_{t-1}\hspace{-0.2pt} &\hspace{-0.2pt} \be_{t}\\ \textbf{0}^\mathrm{T} \hspace{-0.2pt}& \hspace{-0.2pt}\|\bdel_{t}\|_{2} \end{array}\right]= \widetilde{\bU}\widetilde{\bSi}\widetilde{\bV}^\mathrm{T}$.
Eventually, we obtain $\bU_{t}=[\bU_{t-1}~\dfrac{\bdel_{t}}{\|\bdel_{t}\|_{2}}]\cdot\widetilde{\bU}\label{incSVDFinalU}$, $\bSi_{t}\hspace{-0.2pt}=\hspace{-0.2pt}\widetilde{\bSi}$, and $\bV_{t}=\left[\begin{array}{cc} \bV_{t-1}^\mathrm{T}\hspace{-0.2pt} &\hspace{-0.2pt} \textbf{0}\\ \textbf{0}^\mathrm{T}\hspace{-0.2pt} &\hspace{-0.2pt} 1 \end{array}\right]\cdot\widetilde{\bV}$.
\section{Performance Guarantees}\label{performGuarantees}



A bound is presented for the minimum number of measurements required for successful separation by solving the problem \eqref{CORPCAminimization}. This will also serve as a bound for the
proposed CORPCA method. We assume that $\bv_{t}$ is either fixed or slowly-changing
incurring an measurement error that is bounded by $\sigma$, i.e., $\|\mathbf{\Phi}(\hat{\bv}_{t}\hspace{-0.2pt}-\hspace{-0.2pt}\bv_{t})\|_{2}\hspace{-0.2pt}<\hspace{-0.2pt}\sigma$, where $\hat{\bv}_{t}$ is any recovered background. Let $s_0$ denote the support of the source $\bx$ and $s_{j}$ denote the support of each difference vector $\bx\hspace{-0.1pt}-\hspace{-0.1pt}\bz_{j}$; namely,~$\|\bx\hspace{-0.1pt}-\hspace{-0.1pt}\bz_{j}\|_0\hspace{-0.2pt}=\hspace{-0.2pt}s_{j}$, where $j\in\{0,\dots, J\}$. We shall derive the bound for recovering foreground $\bx_{t}$ given either fixed $\bv_{t}$ or the error bound $\sigma$.
\begin{theorem}\label{RAMSIABound}
	The problem in \eqref{CORPCAminimization} requires $m_{t}$ measurements to successfully separate $\bx_{t}\in \mathbb{R}^{n}$ with the sparse support $s_{0}$ and $\bv_{t}\in \mathbb{R}^{n}$ given compressive measurements $\by_{t}=\mathbf{\Phi}(\bx_{t}+\bv_{t})$, where  $\mathbf{\Phi}\in \mathbb{R}^{m_{t}\times n}$, whose
	elements are randomly drawn from an i.i.d. Gaussian distribution, and prior information including the prior $\bZ_{t-1}:=\{\bz_{j}\}_{j=1}^{J}$, where $\|\bx_{t}-\bz_{j}\|_0=s_{j}$ with $\bz_{j}\in \mathbb{R}^{n}$, and the prior $\bB_{t-1}\in \mathbb{R}^{n\times d}$.
	\begin{enumerate}
		\item 
		If $\bv_{t}$ is not changing in time, we can recover $\bx_{t}$ successfully from \eqref{CORPCAminimization} provided
		\begin{equation}\label{boundNoiseless}
		m_{t}\geq 2\alpha\log\Big(\dfrac{n}{\sum_{j=0}^{J}\beta_{j}s_{j}}\Big)\hspace{-0.1pt}+\hspace{-0.1pt}	\dfrac{7}{5}\Big(\sum_{j=0}^{J}\beta_{j}s_{j}\Big)\hspace{-0.1pt}+1.
		\end{equation}	
		
		\item
		If $\bv_{t}$ is slowly-changing, i.e, $\bB_{t-1}$ and $[\bB_{t-1}~\bv_{t}]$ are low rank. We solve the convex problem in  \eqref{CORPCAminimization} and 
		let $\hat{\bx}_{t}$ and $\hat{\bv}_{t}$ denote its solutions. Supposing that the measurement error of $\bv_{t}$ is bounded as $\|\mathbf{\Phi}(\hat{\bv}_{t}-\bv_{t})\|_{2}<\sigma$, we have that $\|\hat{\bx}_{t}-\bx_{t}\|_{2}\leq\frac{2\sigma}{1-\sqrt{\rho}}$ with $0<\rho<1$ provided 
		\begin{equation}\label{boundNoisy}
		m_{t}\geq 2\dfrac{\alpha}{\rho}\log\Big(\dfrac{n}{\sum_{j=0}^{J}\beta_{j}s_{j}}\Big)\hspace{-0.1pt}+\hspace{-0.1pt}	\dfrac{7}{5\rho}\Big(\sum_{j=0}^{J}\beta_{j}s_{j}\Big)\hspace{-0.1pt}+\dfrac{3}{2\rho},
		\end{equation}
	\end{enumerate}
	where $\alpha=\dfrac{\epsilon^{2}}{\hat{\eta}^{2}}\hspace{-0.2pt}\sum\limits_{j=0}^{J}\hspace{-0.2pt}\beta_{j}\hspace{-0.2pt}\sum\limits_{i=1}^{s_{j}}w_{ji}^{2}$ in both \eqref{boundNoiseless} and \eqref{boundNoisy}, wherein $\epsilon>0$ and $\hat{\eta}=\min\limits_{j}\Big\{n\Big(\sum\limits_{i=1}^{n}\frac{1}{|x_{ti}\hspace{-0.1pt}-\hspace{-0.1pt}z_{ji}|+\epsilon} \Big)^{-1}\Big\}$.
\end{theorem}
\begin{proof}\label{RAMSIABoundProof}
We firstly derive the bound in \eqref{boundNoiseless} of Theorem \ref{RAMSIABound} for which the low-rank component $\bv_{t}$ is not changed and the remaining bound is a noisy case when the low-rank component $\bv_{t}$ is slowly-changing. We drive the bounds based on Proposition \ref{propUpper} (see Appendix \ref{propositionAppendix}) by first computing the subdifferential $\partial g(\bx)$ and then the distance, $\mathrm{dist}(\cdot)$, between the standard normal vector $\bg$ and $\partial g(\bx)$. The $\bu \in \partial g(\bx)$ of $g(\bx)$ is derived through the separate components of the sum $g(\bx)\hspace{-0.2pt}=\hspace{-0.2pt}\sum_{j=0}^{J}\beta_{j}g_{j}(\bx)$, where $g_{j}(\bx)\hspace{-0.2pt}=\hspace{-0.2pt}\|\mathbf{W}_{j}(\bx-\bz_{j})\|_{1}$. As a result, $\partial g(\bx)=\sum_{j=0}^{J}\beta_{j}\partial g_{j}(\bx)$.

Considering the distance from the standard normal vector $\bg$ to the subdifferential $\partial g(\bx)$ given by \eqref{euclideanDistance}, we derive a relation between this distance and all separate distances of $g_{j}(\bx)$ as
\begin{equation}\label{distanceToSeparateDistance}
\mathrm{dist}^{2}\hspace{-0.2pt}(\bg,\tau\hspace{-0.2pt} \cdot\hspace{-0.2pt} \partial g(\bx))\hspace{-0.2pt}
\leq \sum_{j=0}^{J}\beta_{j}\mathrm{dist}^{2}\hspace{-0.2pt}(\bg,\tau\hspace{-0.2pt} \cdot\hspace{-0.2pt}\partial g_{j}(\bx)),
\end{equation}
where $\sum_{j=0}^{J}\beta_{j}\hspace{-0.2pt}=\hspace{-0.2pt}1$. Taking the expectation of \eqref{distanceToSeparateDistance} delivers
\begin{equation}\label{distanceToSubdifferentialExpectationIS}
\mathbb{E}_{\bg}[\mathrm{dist}^{2}\hspace{-0.1pt}(\bg,\tau\hspace{-0.1pt} \cdot\hspace{-0.1pt} \partial g(\bx))]\hspace{-0.1pt}\leq
\sum_{j=0}^{J}\beta_{j}\mathbb{E}_{\bg}[\mathrm{dist}^{2}\hspace{-0.1pt}(\bg,\tau\hspace{-0.1pt} \cdot\hspace{-0.1pt}\partial g_{j}(\bx))]\hspace{-0.1pt}.
\end{equation}

\subsubsection{Distance Expectation}
We now consider each component $\mathbb{E}_{\bg}[\mathrm{dist}^{2}\hspace{-0.2pt}(\bg,\tau\hspace{-0.1pt} \cdot\hspace{-0.1pt}\partial g_{j}(\bx))]$, where $g_{j}(\bx)\hspace{-0.2pt}=\hspace{-0.2pt}\|\mathbf{W}_{j}(\bx\hspace{-0.2pt}-\hspace{-0.2pt}\bz_{j})\|_{1}$. We firstly consider $\bx\hspace{-0.1pt}-\hspace{-0.1pt}\bz_{j}$. Without loss of generality, we have assumed that the sparse vector $\bx\hspace{-0.2pt}-\hspace{-0.2pt}\bz_{j}$ has $s_{j}$ supports or $\|\bx-\bz_{j}\|_{0}\hspace{-0.2pt}=\hspace{-0.2pt}s_{j}$, i.e, taking the form $\bx\hspace{-0.2pt}-\hspace{-0.2pt}\bz_{j}\hspace{-0.2pt} =\hspace{-0.2pt} (x_{1}\hspace{-0.2pt}-\hspace{-0.2pt}z_{j1},...,x_{s_{j}}\hspace{-0.2pt}-\hspace{-0.2pt}z_{js_{j}},0,...,0)$.
We have the subdifferential $\bu \hspace{-0.2pt}\in \hspace{-0.2pt}\partial g_{j}(\bx)$ of $g_{j}(\bx)$ given by:
\begin{equation}\label{wl1Subdifferential}
\left\{
\begin{array}{l}
u_{i} ~~= w_{ji}\mathrm{sign}(x_{i}\hspace{-0.2pt}-\hspace{-0.2pt}z_{ji}),~~~~~~i\hspace{-0.2pt}=\hspace{-0.2pt}1,...,s_{j} \\
|u_{i}| \leq w_{ji},~~~~~~~~~~~~~~~~~~~~~~i\hspace{-0.2pt}=\hspace{-0.2pt}s_{j}\hspace{-0.2pt}+\hspace{-0.2pt}1,...,n.
\end{array}
\right.
\end{equation}

Let us recall the weights as proposed in RAMSIA \cite{LuongICIP16} for a specific $\bz_{j}$ with $w_{ji}\hspace{-0.2pt}=\hspace{-0.2pt}\eta_{j}/(|x_{i}\hspace{-0.2pt}-\hspace{-0.2pt}z_{ji}|\hspace{-0.2pt}+\epsilon)$ and the constraint $\sum_{i=1}^{n}w_{ji}\hspace{-0.2pt}=\hspace{-0.2pt}n$. Taking the source $\bx\hspace{-0.2pt}-\hspace{-0.2pt}\bz_{j}$ to derive
\begin{equation}\label{wl1BoundWeights}
w_{ji}\hspace{-0.2pt}=\hspace{-0.2pt}\left\{
\begin{array}{l}
\frac{\eta_{j}}{|x_{i}\hspace{-0.1pt}-\hspace{-0.1pt}z_{ji}|\hspace{0pt}+\epsilon},~~~~~~i\hspace{-0.2pt}=\hspace{-0.2pt}1,...,s_{j} \\
\frac{\eta_{j}}{\epsilon},~~~~~~~~~~~~~~i\hspace{-0.2pt}=\hspace{-0.2pt}s_{j}\hspace{-0.2pt}+\hspace{-0.2pt}1,...,n.
\end{array}
\right.
\end{equation}
From the constraint $\sum_{i=1}^{n}w_{ji}\hspace{-0.2pt}=\hspace{-0.2pt}n$, we can derive: $\eta_{j} =n\Big(\sum\limits_{i=1}^{n}\frac{1}{|x_{i}\hspace{0pt}-\hspace{0pt}z_{ji}|+\epsilon} \Big)^{-1}$.

We can then compute the distance from the standard normal vector $\bg$ to the subdifferential $\partial g_{j}(\bx)$ based on \eqref{euclideanDistance} as
\begin{equation}\label{distanceToSubdifferential}
\begin{split}
\mathrm{dist}^{2}(\bg,\tau \cdot \partial g_{j}(\bx))\hspace{-0.2pt}&= \hspace{-0.5pt}\sum\limits_{i=1}^{s_{j}}\Big(\g_{i}\hspace{-0.2pt}-\tau w_{ji}\mathrm{sign}(x_{i}\hspace{-0.2pt}-\hspace{-0.2pt}z_{ji})\Big)^{2}\hspace{-0.2pt}\\
&+\hspace{-0.2pt}\hspace{-0.5pt}\sum\limits_{i=s_{j}+1}^{n}\hspace{-0.5pt}\Big(\max (|\g_{i}|-\tau w_{ji},0)\Big)^{2},
\end{split}
\end{equation}
where $\max(a,0)$ returns the maximum value between $a\hspace{-0.2pt} \in \hspace{-0.2pt}\mathbb{R}$ and $0$.
Further, after taking the expectation of \eqref{distanceToSubdifferential}, we obtain 
\begin{equation}\label{distanceToSubdifferentialExpectation}
\begin{split}
\mathbb{E}_{\bg}[\mathrm{dist}^{2}(\bg,&\tau \cdot \partial g_{j}(\bx))]=s_{j}\hspace{-0.2pt}+\tau^{2} \hspace{-0.2pt}\sum\limits_{i=1}^{s_{j}}w_{ji}^{2}
\\
&+\sqrt{\frac{2}{\pi}}\hspace{-0.2pt}\sum\limits_{i=s_{j}+1}^{n}\int_{\tau w_{ji}}^{\infty}\hspace{-0.2pt}(v-\tau w_{ji})\mathrm{e}^{-v^{2}\hspace{-0.2pt}/2}\mathrm{d}v.
\end{split}
\end{equation}
We use the inequality \eqref{positiveA} for the third term in \eqref{distanceToSubdifferentialExpectation} to get
\begin{equation}\label{distanceToSubdifferentialExpectationLE}
\mathbb{E}_{\bg}[\mathrm{dist}^{2}(\bg,\tau \cdot \partial g_{j}(\bx))]\hspace{-0.2pt}\leq \hspace{-0.2pt}s_{j}\hspace{-0.1pt}+\tau^{2} \hspace{-0.2pt}\sum\limits_{i=1}^{s_{j}}\hspace{-0.1pt}w_{ji}^{2}
\hspace{-0.1pt}+\hspace{-0.1pt}2\hspace{-0.6pt}\sum\limits_{i=s_{j}+1}^{n}\hspace{-0.6pt}\frac{\psi(\tau w_{ji})}{\tau w_{ji}}.
\end{equation}
\subsubsection{Bound Derivation}
Inserting \eqref{distanceToSubdifferentialExpectationLE} in \eqref{distanceToSubdifferentialExpectationIS} for all functions $g_{j}(\bx)$ gives
\begin{equation}\label{distanceToSubdifferentialExpectationLEAll}
\begin{split}
\mathbb{E}_{\bg}[\mathrm{dist}^{2}&(\bg,\tau \cdot \partial g(\bx))]\leq \sum\limits_{j=0}^{J}\beta_{j}s_{j}\hspace{-0.2pt}+\tau^{2} \hspace{-0.2pt}\sum\limits_{j=0}^{J}\beta_{j}\sum\limits_{i=1}^{s_{j}}w_{ji}^{2}
\hspace{-0.2pt}\\
&+\hspace{-0.2pt}2\sum\limits_{j=0}^{J}\hspace{-0.2pt}\beta_{j}(n-s_{j})\frac{\psi(\tau \eta_{j}/\epsilon)}{\tau \eta_{j}/\epsilon},
\end{split}
\end{equation}
where $w_{ji}\hspace{-0.2pt}=\hspace{-0.2pt}\eta_{j}/\epsilon$ in \eqref{wl1BoundWeights} are used in the third term of \eqref{distanceToSubdifferentialExpectationLEAll}. Let us denote $\widehat{\eta}=\min\limits_{j}\{\eta_{j}\}$. It can be noted that the function $\psi(x)/x$ decreases with $x\hspace{-0.2pt}>\hspace{-0.2pt}0$. Thus $\frac{\psi(\tau \eta_{j}/\epsilon)}{\tau \eta_{j}/\epsilon}\leq \frac{\psi(\tau \widehat{\eta}/\epsilon)}{\tau \widehat{\eta}/\epsilon}$. Considering the last term in \eqref{distanceToSubdifferentialExpectationLEAll}, we derive
\begin{equation}\label{distanceToSubdifferentialExpectationLEAllMin}
\begin{split}
\mathbb{E}_{\bg}[\mathrm{dist}^{2}(\bg,\tau \cdot \partial g(\bx))]&\leq \sum\limits_{j=0}^{J}\beta_{j}s_{j}\hspace{-0.2pt}+\tau^{2} \hspace{-0.2pt}\sum\limits_{j=0}^{J}\beta_{j}\sum\limits_{i=1}^{s_{j}}w_{ji}^{2}
\hspace{-0.2pt}\\
&+\hspace{-0.2pt}2\Big(n-\sum\limits_{j=0}^{J}\hspace{-0.2pt}\beta_{j}s_{j}\Big)\frac{\psi(\tau \widehat{\eta}/\epsilon)}{\tau \widehat{\eta}/\epsilon}.
\end{split}
\end{equation}
Inequality \eqref{distanceToSubdifferentialExpectationLEAllMin} holds due to $\sum_{j=0}^{J}\hspace{-0.2pt}\beta_{j}(n-s_{j})=n-\sum_{j=0}^{J}\hspace{-0.2pt}\beta_{j}s_{j}$.

We denote $\bar{s}\hspace{-0.2pt}=\hspace{-0.2pt}\sum_{j=0}^{J}\beta_{j}s_{j}$ and from \eqref{upperBoundCompute}, \eqref{distanceToSubdifferentialExpectationLEAllMin}, we have
\begin{equation}\label{upperBoundComputeWL1}
U_{g}\hspace{-0.3pt}= \min_{\tau\geq 0}\Big\{\bar{s}\hspace{-0.1pt}+\hspace{-0.2pt}\tau^{2} \hspace{-0.2pt}\sum\limits_{j=0}^{J}\hspace{-0.2pt}\beta_{j}\hspace{-0.2pt}\sum\limits_{i=1}^{s_{j}}w_{ji}^{2}
\hspace{-0.2pt}
+\hspace{-0.2pt}2(n\hspace{-0.2pt}-\hspace{-0.2pt}\bar{s})\frac{\psi(\tau \widehat{\eta}/\epsilon)}{\tau \widehat{\eta}/\epsilon}\Big\}.
\end{equation}
Inserting \eqref{densityDistribution} in \eqref{upperBoundComputeWL1} gives
\begin{equation}\label{upperBoundComputeWL1ApplyAEta}
\hspace{-0.1pt}U_{g}\hspace{-0.3pt}\leq \hspace{-0.2pt}\min_{\tau\geq 0}\hspace{-0.2pt}\Big\{\hspace{-0.2pt}\bar{s}\hspace{-0.1pt}+\hspace{0pt}\tau^{2} \hspace{-0.2pt}\sum\limits_{j=0}^{J}\hspace{-0.2pt}\beta_{j}\hspace{-0.2pt}\sum\limits_{i=1}^{s_{j}}\hspace{-0.1pt}w_{ji}^{2}
\hspace{-0.1pt}
+\hspace{-0.1pt}(n\hspace{-0.1pt}-\hspace{-0.1pt}\bar{s})\frac{2}{\sqrt{2\pi}}\frac{e^{(-\tau^{2}\widehat{\eta}^{2}/2\epsilon^{2})}}{\tau\widehat{\eta}/\epsilon}\hspace{-0.2pt}\Big\}.    
\end{equation}

To give a bound as a function of a given source $\bx$ and other parameters, we can select a parameter $\tau\hspace{-0.2pt}>\hspace{-0.2pt}0$ to obtain an useful bound in \eqref{upperBoundComputeWL1ApplyAEta}. Setting $\tau\hspace{-0.2pt}=\hspace{-0.2pt}(\epsilon/\widehat{\eta})\sqrt{2\log(n/\bar{s})}$ gives
\begin{equation}\label{upperBoundComputeWL1ApplyAEtaSet}
U_{g}\hspace{-0.3pt}\leq \bar{s}+2\log(\frac{n}{\bar{s}})\frac{\epsilon^{2}}{\widehat{\eta}^{2}}\hspace{-0.2pt}\sum\limits_{j=0}^{J}\hspace{-0.2pt}\beta_{j}\hspace{-0.2pt}\sum\limits_{i=1}^{s_{j}}w_{ji}^{2}
+\frac{\bar{s}(1-\bar{s}/n)}{\sqrt{\pi\log(n/\bar{s})}}.     \\
\end{equation}
We denote $\alpha\hspace{-0.1pt}=\hspace{-0.1pt}\frac{\epsilon^{2}}{\widehat{\eta}^{2}}\hspace{-0.1pt}\sum_{j=0}^{J}\hspace{-0.1pt}\beta_{j}\hspace{-0.1pt}\sum_{i=1}^{s_{j}}w_{ji}^{2}
\hspace{-0.1pt}$ in the second term in \eqref{upperBoundComputeWL1ApplyAEtaSet}.
Finally, applying the inequality \eqref{logInequality} to the last term in \eqref{upperBoundComputeWL1ApplyAEtaSet} gives 
\begin{equation}\label{upperBoundComputeWL1ApplyAEtaSetLog}
U_{g}\hspace{-0.3pt}\leq 2\alpha\log(\frac{n}{\bar{s}})\hspace{-0.2pt}+\frac{7}{5}\bar{s}.
\end{equation}
As a result, from Proposition \ref{propUpper} (in Appendix \ref{appendix}), we get the bound in \eqref{boundNoiseless} as
\begin{equation}\label{upperBoundComputeWL1ApplyAEtaSetMeasurement}
m_{t}\geq 2\alpha\log(\frac{n}{\bar{s}})\hspace{-0.2pt}+\frac{7}{5}\bar{s}+1.
\end{equation}
For slowly-changing $\bv_{t}$, we assume that the measurement error $\|\mathbf{\Phi}(\hat{\bv}_{t}-\bv_{t})\|_{2}<\sigma$. Applying Proposition \ref{propUpper} (see Appendix \ref{appendix}) for the noisy case, we get the bound in \eqref{boundNoisy}.
\end{proof}

%

\section{Numerical Experiment}
\label{Experiment}
\begin{figure*}[tp!]
	\centering
	\subfigure[$n$-$\ell_{1}$ minimization]{
		\includegraphics[width=0.25\textwidth]{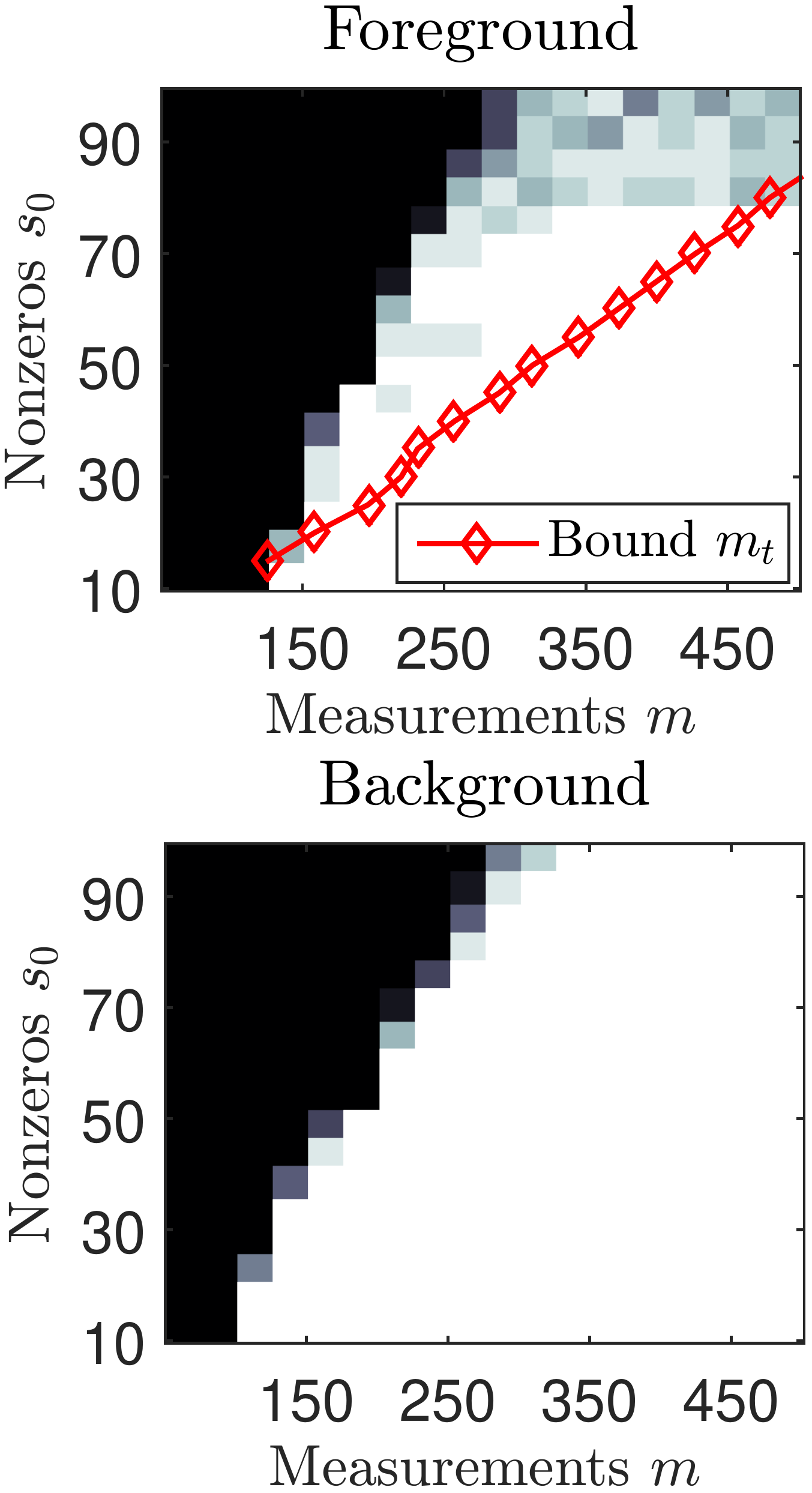}\label{corpcaPerform}}
	\subfigure[$\ell_{1}$-$\ell_{1}$ minimization]{
		\includegraphics[width=0.25\textwidth]{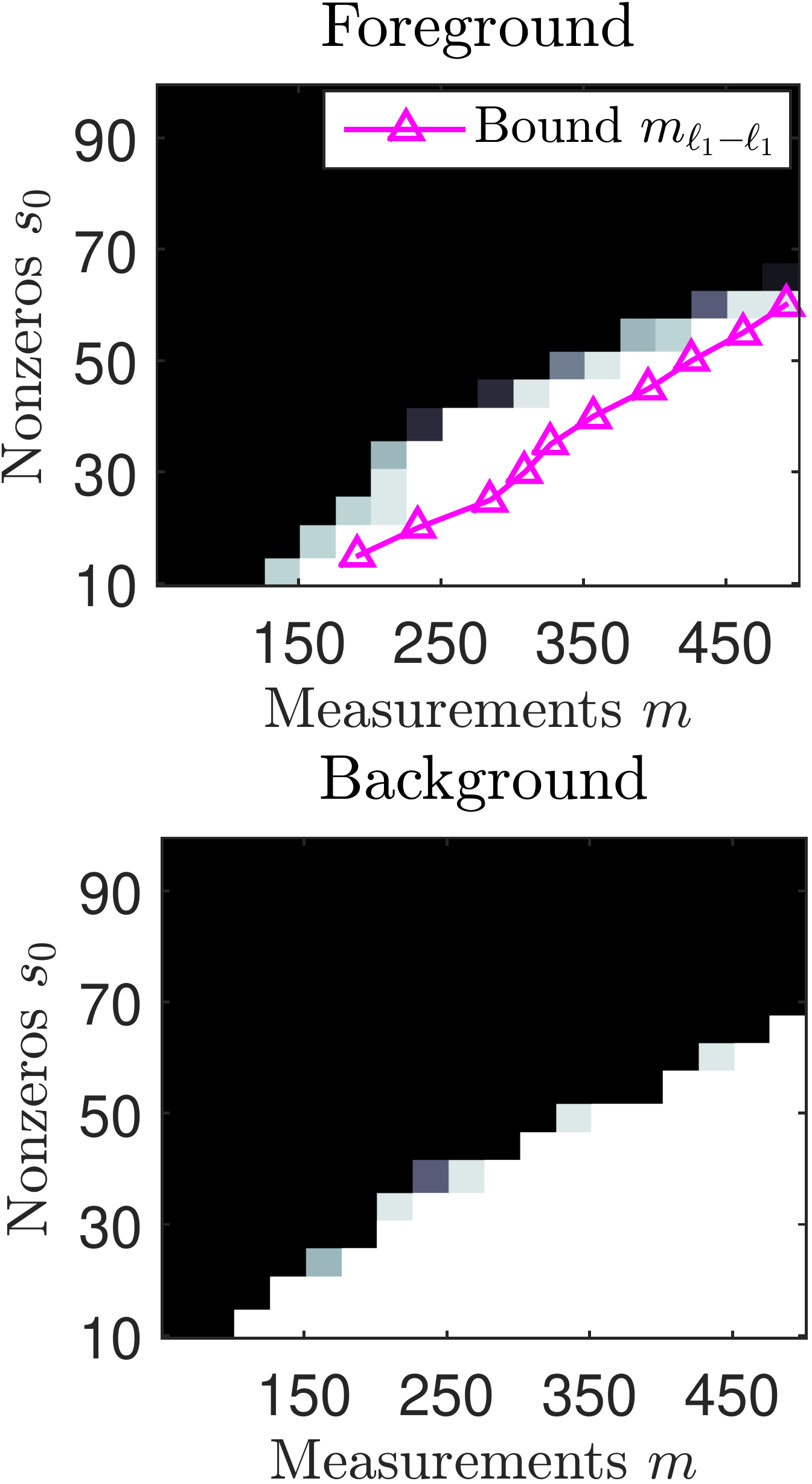}\label{fistaL1L1Perform}}
	\subfigure[$\ell_{1}$ minimization]{
		\includegraphics[width=0.25\textwidth]{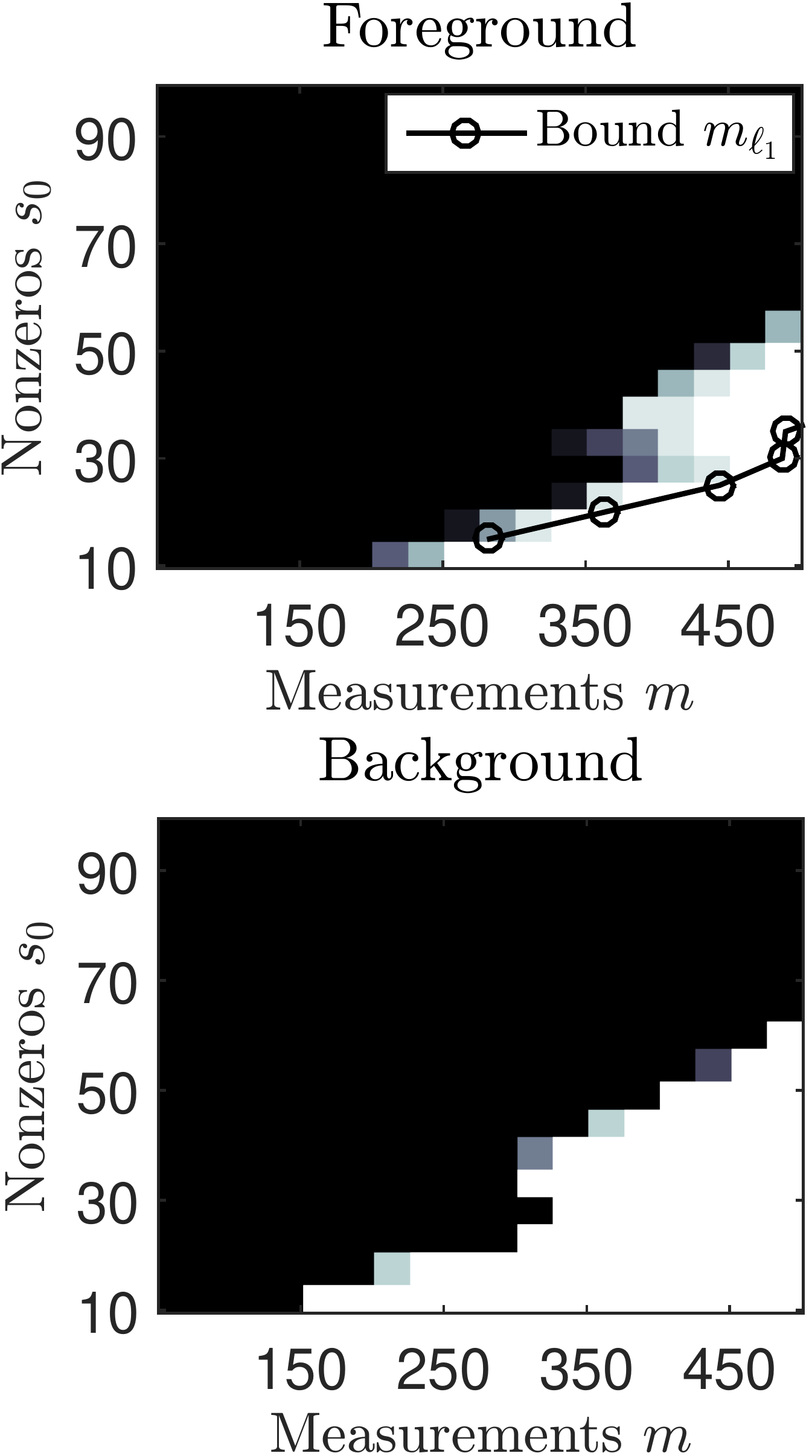}\label{fistaPerform}}	 
	\subfigure[Scale]{
		\includegraphics[width=0.11\textwidth]{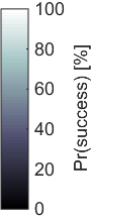}\label{scale}}		
	\caption{The averages of success probabilities and bounds of CORPCA with different minimizations (a)-(c), to recover $\bx_{t}\in\mathbb{R}^{500}$, where the $x$-axis is the number of measurements $m$ and the $y$-axis is the number of nonzeros $s_{0}$. The gray scale (d) is proportional to $\mathrm{Pr(sussess)}$ from pure black to white.
	}\label{figAccuracy}
\end{figure*}
We evaluate the performance of our Algorithm \ref{CORPCAAlg} employing the proposed method, the existing $\ell_{1}$ minimization \cite{Beck09}, and the existing $\ell_{1}$-$\ell_{1}$ \cite{MotaGLOBALSIP14,MotaARXIV14} minimization methods. We also compute the established bound \eqref{boundNoisy} in Theorem \ref{RAMSIABound}, which serves as an estimate of number of measurements that CORPCA requires for successful recovery. We generate our data as follows. We generate the low-rank component $\bL\hspace{-0.2pt}=\hspace{-0.2pt}\bU\bV^{\mathrm{T}}$, where $\bU\hspace{-0.2pt}\in\hspace{-0.2pt}\mathbb{R}^{n\times r}$ and $\bV\hspace{-0.2pt}\in\hspace{-0.2pt}\mathbb{R}^{(d+q)\times r}$ are random matrices whose entries are drawn from the standard normal distribution. We set $n\hspace{-0.2pt}=\hspace{-0.2pt}500$, $r\hspace{-0.1pt}=\hspace{-0.1pt}5$ (rank of $\bL$), and that $d\hspace{-0.1pt}=\hspace{-0.1pt}100$ is the number of vectors for training and $q=100$ is the number of testing vectors. This yields $\bL=[\bx_{1}~...~\bx_{d+q}]$. We generate $\bS=[\bv_{1}~...~\bv_{d+q}]$, where at time instance $t-1$, $\bx_{t-1}\in\mathbb{R}^{n}$ is generated from
the standard normal distribution with support $s_{0}$. Our purpose is to consider a sequence of correlated sparse vectors $\bx_{t}$. Therefore, we generate $\bx_{t}$ satisfying $\|\bx_{t}-\bx_{t-1}\|_{0}=s_{0}/2$. This could lead to $\|\bx_{t}\|_{0}> s_{0}$. To avoid a large increase of $\|\bx_{t}\|_{0}$, we set the constraint $\|\bx_{t}\|_{0}\in[s_{0},s_{0}+15]$, whenever $\|\bx_{t}\|_{0}>s_{0}+15$, $\bx_{t}$ is randomly reset to  $\|\bx_{t}\|_{0}=s_{0}$ by setting $\|\bx_{t}\|_{0}-s_{0}$ positions that are randomly selected to zero. In this work, we test our algorithms for $s_{0}=10$ to 90.

Our prior information is initialized as follows. To illustrate real scenarios, where we do not know the sparse and low-rank components, we use the batch-based RPCA \cite{CandesRPCA} to separate the training set $\bM_{0}=[\bx_{1}+\bv_{1}~...~\bx_{d}+\bv_{d}]$ to obtain $\bB_{0}=[\bv_{1}~...~\bv_{d}]$. In this experiment, we use three ($J=3$) sparse components as prior information and we set $\bZ_{0}:=\{\mathbf{0},\mathbf{0},\mathbf{0}\}$. We run CORPCA (Sec. \ref{corpca}) on the test set $\bM=[\bx_{d+1}+\bv_{d+1}~...~\bx_{d+q}+\bv_{d+q}]$. 

We assess the accuracy of recovery $\bhx_{t}$ versus $\bx_{t}$ in terms of the success probability, denoted as $\mathrm{Pr(success)}$, versus the number of measurements $m$. At time instance $t$, the original source $\bx_{t}\hspace{-0.1pt}+\hspace{-0.1pt}\bz_{t}$ is projected onto lower-dimensional observations $\by_{t}$ having $m$ measurements. For a fixed $m$, the $\mathrm{Pr(success)}$ is the number of times, in which the source $\bx_{t}$ is recovered as $\bhx_{t}$ with an error $\|\bhx_{t}-\bx_{t}\|_{2}/\|\bx_{t}\|_{2}\leq \hspace{-0.2pt}10^{-2}$, divided by the total 50 Monte Carlo simulations and where we have set parameters $\epsilon\hspace{-1pt}=\hspace{-1pt}0.8$, $\lambda\hspace{-1pt}=\hspace{-1pt}1/\sqrt{n}$. In addition, we will evaluate the obtained bound $m_t$ \eqref{boundNoisy}, $m_{\ell_{1}\hspace{-0.1pt}\text{-}\ell_{1}}$ \eqref{l1-l1 bound}, and $m_{\ell_{1}}$ \eqref{l1 bound} \cite{MotaGLOBALSIP14} in the presence of noisy measurements $\|\mathbf{\Phi}(\hat{\bv}_{t}\hspace{-0.1pt}-\hspace{-0.1pt}\bv_{t})\|_{2}\hspace{-0.2pt}<\hspace{-0.2pt}\sigma$ in the case of slowly-changing low-rank components $\bv_{t}$. As shown in \eqref{boundNoisy} of Theorem \ref{RAMSIABound}, these bounds depend on $\rho$. We experimentally select $\rho\hspace{-0.2pt}\in\hspace{-0.2pt}(0,1)$ to give that our estimated bound is consistent with the practical results, here we have set $\rho$ to values  of $\frac{0.8}{3};\frac{0.6}{3};\frac{0.4}{3};~\text{and}$ for $m_t$ $m_{\ell_{1}\hspace{-1pt}\text{-}\ell_{1}}$, and $m_{\ell_{1}}$, respectively.

The results in Fig. \ref{corpcaPerform} show the efficiency of CORPCA employing $n$-$\ell_{1}$ minimization: that, at certain sparsity degrees, we can recover the 500-dimensional data from measurements of much lower dimesions ($m = 150$ to 300, see white areas in Fig. \ref{corpcaPerform}). Furthermore, the measurement bound $m_{t}$ (red line in Fig. \ref{corpcaPerform}) is consistent with the practical results. It is also clear that the $\ell_{1}$ and $\ell_{1}$-$\ell_{1}$ minimization methods (see Figs. \ref{fistaL1L1Perform}+\ref{fistaPerform}) lead to a higher number of measurements, thereby illustrating the benefit of incorporating multiple side information into the problem.
\vspace{-0pt}
\section{Video Foreground-Background Separation}
We assess our CORPCA method in the application of compressive video separation and compare it against existing methods. We run all methods listed in Table \ref{tableMethods} on typical test video content \cite{Li04}. 
In this experiment, we use $d=100$ frames as training vectors for the proposed CORPCA as well as for GRASTA \cite{JHe12} and ReProCS \cite{GuoQV14}.
\vspace{-1pt}
\vspace{-1pt}
\begin{table}[t!]
	\setlength{\tabcolsep}{5pt}
	\renewcommand{\arraystretch}{1.3}       
	\caption{Summary of recovery characteristics}
	\label{tableMethods}
	\vspace{-0pt}
	\centering
	\begin{tabular}{l c c c c c}
		\hline                
		\hspace{-3pt}&&\hspace{-10pt}CORPCA &\hspace{-3pt}RPCA &\hspace{-3pt}GRASTA& \hspace{-3pt}ReProCS\\  
		\vspace{-1pt}   
		\hspace{-5pt}
		&\hspace{-5pt}&\hspace{-5pt}  &\hspace{-5pt}\cite{CandesRPCA} &\hspace{-5pt}\cite{JHe12} &\hspace{-5pt}\cite{GuoQV14}\\       
		\hline
		\hspace{-3pt}Online &&\hspace{-10pt} \color{blue}{\checkmark} & & \hspace{-3pt}\color{blue}{\checkmark}&\hspace{-3pt}\color{blue}{\checkmark}\\
		\hline
		\hspace{-3pt}Full data&&  \hspace{-10pt}\color{blue}{\checkmark} &  \hspace{-3pt}\color{blue}{\checkmark} &  \hspace{-3pt}\color{blue}{\checkmark}& \hspace{-3pt}\color{blue}{\checkmark}\\
		\hline
		\hspace{-3pt}\multirow{2}{*}{Compressed}\hspace{-3pt}&\hspace{-3pt}Foreground\hspace{-3pt}&  \hspace{-10pt}\color{blue}{\checkmark} &   &  & \hspace{-3pt}\color{blue}{\checkmark}\\
		
		&\hspace{-3pt}Background\hspace{-3pt}&   \hspace{-10pt}\color{blue}{\checkmark} &   & \hspace{-3pt}\color{blue}{\checkmark}&\\
		\hline
		\vspace{-1pt}
	\end{tabular}
\end{table} 
\subsubsection{Visual Evaluation}
\vspace{-0pt}
We first consider background-foreground video separation with full access to the video data (the data set $\bM$); the visual results of the various methods are illustrated in Fig. \ref{figVisualPerform}. It is evident that, for both the video sequences, CORPCA delivers superior visual results than the other methods, which suffer from less-details in the foreground and noisy background images.

Fig. \ref{figVisualCompressed} presents the results of CORPCA under various rates on the number of measurements $m$ over the dimension $n$ of the data (the size of the vectorized frame). The results show that we can recover the foreground and background even by accessing a small number of measurements; for instance, we can obtain good-quality reconstructions with only $m/n=0.6$ and $m/n=0.4$ for \texttt{Bootstrap} [see Fig. \ref{bootstrapCompressed}] and \texttt{Curtain} [see Fig. \ref{curtainCompressed}], respectively. \texttt{Bootstrap} requires more measurements than \texttt{Curtain} due to the more complex foreground information. For comparison, we illustrate the visual results obtained with ReProCS---which, however, can only recover the foreground using compressive measurements---in Fig.~\ref{figVisualCompressedReprocs}. It is clear that the reconstructed foreground images have a poorer visual quality compared to CORPCA even at a high rate $m/n=0.8$\footnote{The original test videos and the reconstructed separated sequences are available online~\cite{corpcaSource}.}. 
\begin{figure}[tp!]
	\centering
	\subfigure[\vspace{-0.15pt}\texttt{Bootstrap}]{
		\includegraphics[width=0.48\textwidth]{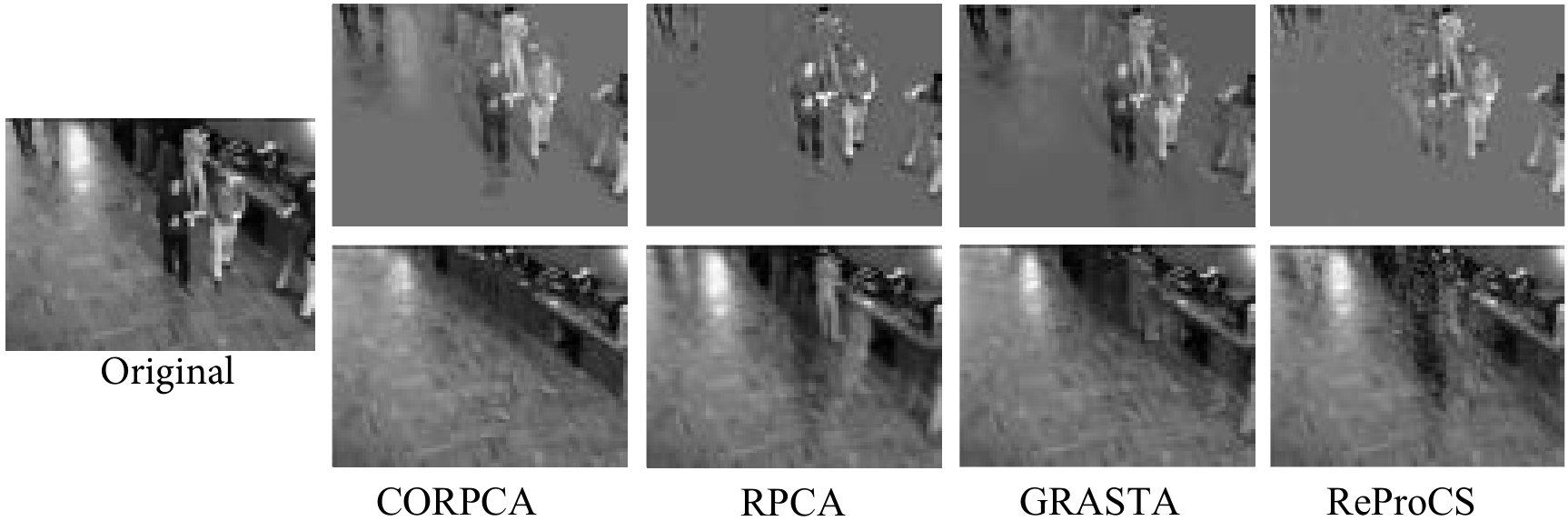}\label{bootstrapRateFull}}
	%
	\subfigure[\vspace{-0.15pt}\texttt{Curtain}]{
		\includegraphics[width=0.48\textwidth]{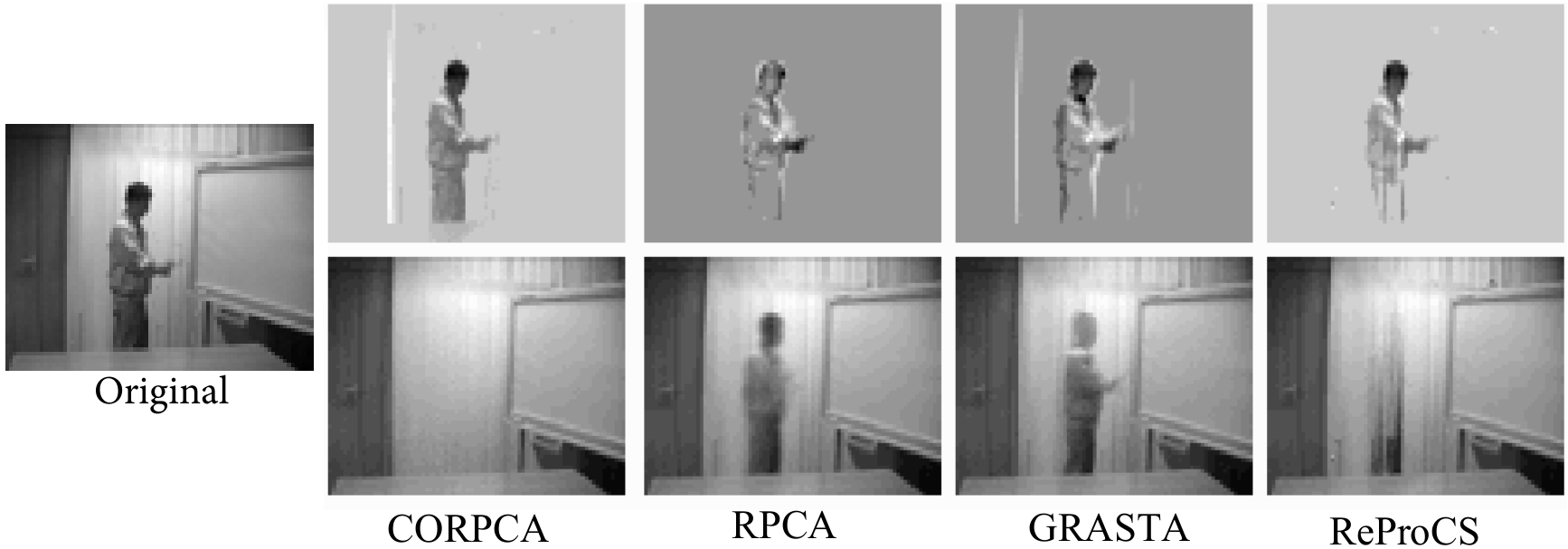}\label{curtainRateFull}}
	\vspace{-4pt}
	\caption{Background and foreground separation for the different separation methods with full data access \texttt{Bootstrap\#2213} and \texttt{Curtain\#2866}.
	}\label{figVisualPerform}
	
	\vspace{-0pt}
\end{figure}
\begin{figure}[tp!]
	\centering
	\subfigure[\vspace{-0.15pt}\texttt{Bootstrap}]{
		\includegraphics[width=0.43\textwidth]{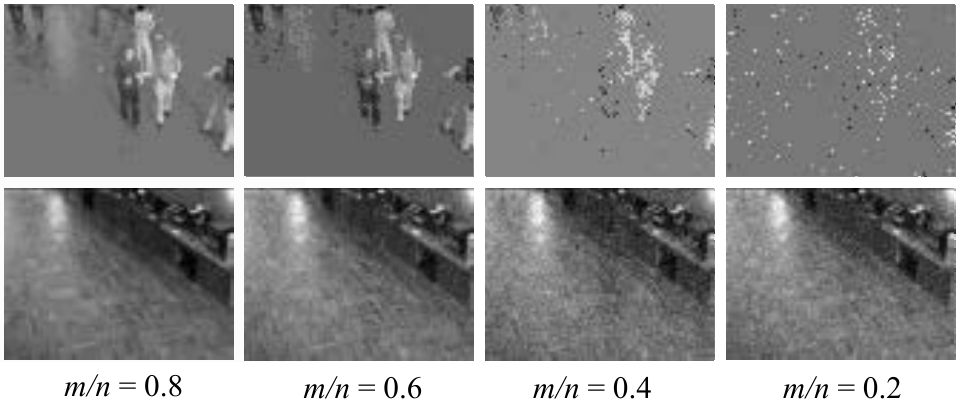}\label{bootstrapCompressed}}
	\subfigure[\vspace{-0.15pt}\texttt{Curtain}]{
		\includegraphics[width=0.43\textwidth]{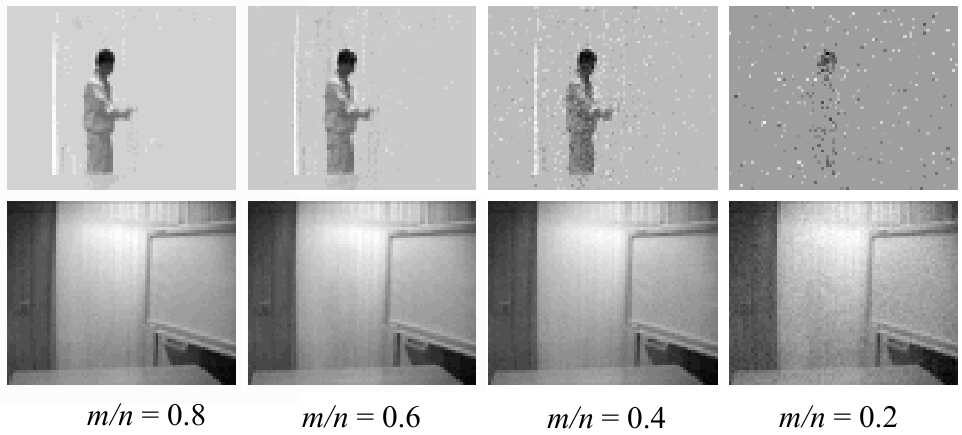}\label{curtainCompressed}}
	\vspace{-4pt}
	\caption{Compressive background and foreground separation of CORPCA with different measurement rates $m/n$.
	}\label{figVisualCompressed}
	\vspace{-0pt}
\end{figure}
\begin{figure}[th!]
	\centering
	\subfigure[\vspace{-0.15pt}\texttt{Bootstrap}]{
		\includegraphics[width=0.43\textwidth]{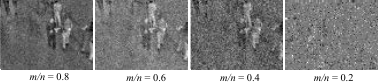}\label{bootstrapCompressedReprocs}}
	\subfigure[\vspace{-0.15pt}\texttt{Curtain}]{
		\includegraphics[width=0.43\textwidth]{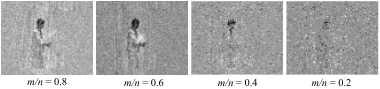}\label{curtainCompressedReprocs}}
	\vspace{-4pt}
	\caption{Compressive foreground separation of ReProCS with different measurement rates $m/n$.
	}\label{figVisualCompressedReprocs}
	\vspace{-0pt}
\end{figure}
We consider two videos \cite{Li04}, \texttt{Bootstrap} (60$\times$80 pixels) and \texttt{Curtain} (64$\times$80 pixels) [c.f., Fig. \ref{figVisualPerform}], having a static and a dynamic background, respectively.

\subsubsection{Quantitative Results}
\vspace{-0pt}
\begin{figure*}[tp!]
	\centering
	\subfigure[\vspace{-0.15pt}\texttt{Bootstrap}: CORPCA-$n$-$\ell_1$]{
		\includegraphics[width=0.315\textwidth]{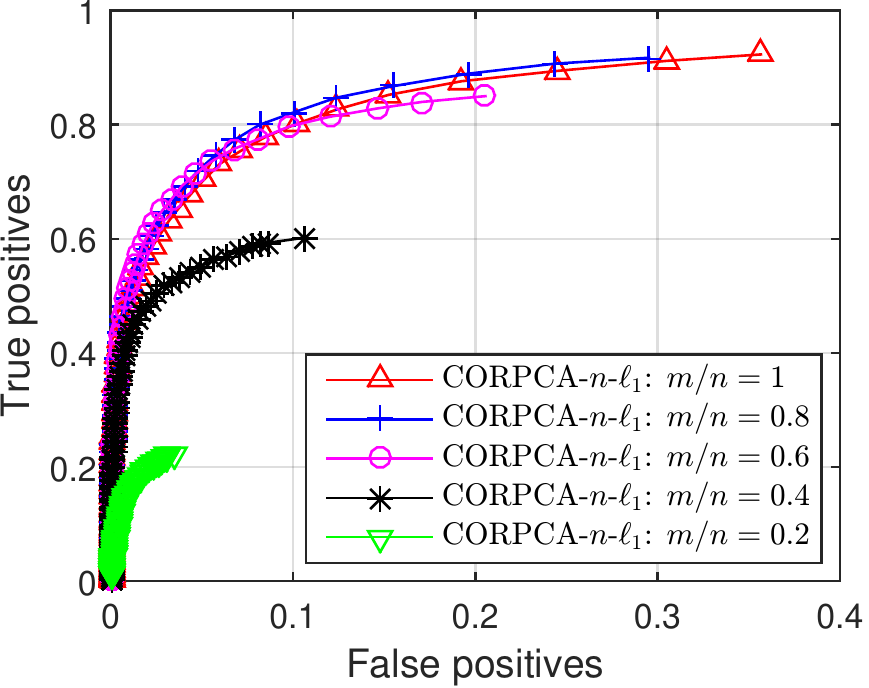}\label{bootstrapROCCorpcanl1}}
	\subfigure[\vspace{-0.15pt}\texttt{Bootstrap}: CORPCA-$\ell_1$-$\ell_1$]{
		\includegraphics[width=0.325\textwidth]{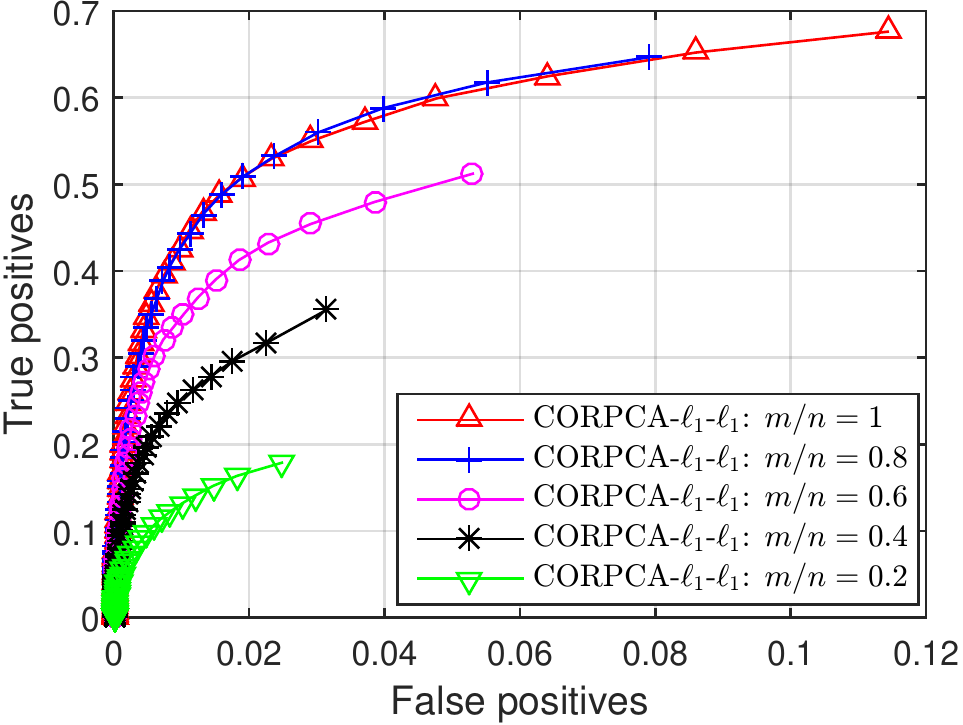}\label{bootstrapROCCorpcal1l1}}	
	\subfigure[\vspace{-0.15pt}\texttt{Bootstrap}: CORPCA-$\ell_1$]{
		\includegraphics[width=0.325\textwidth]{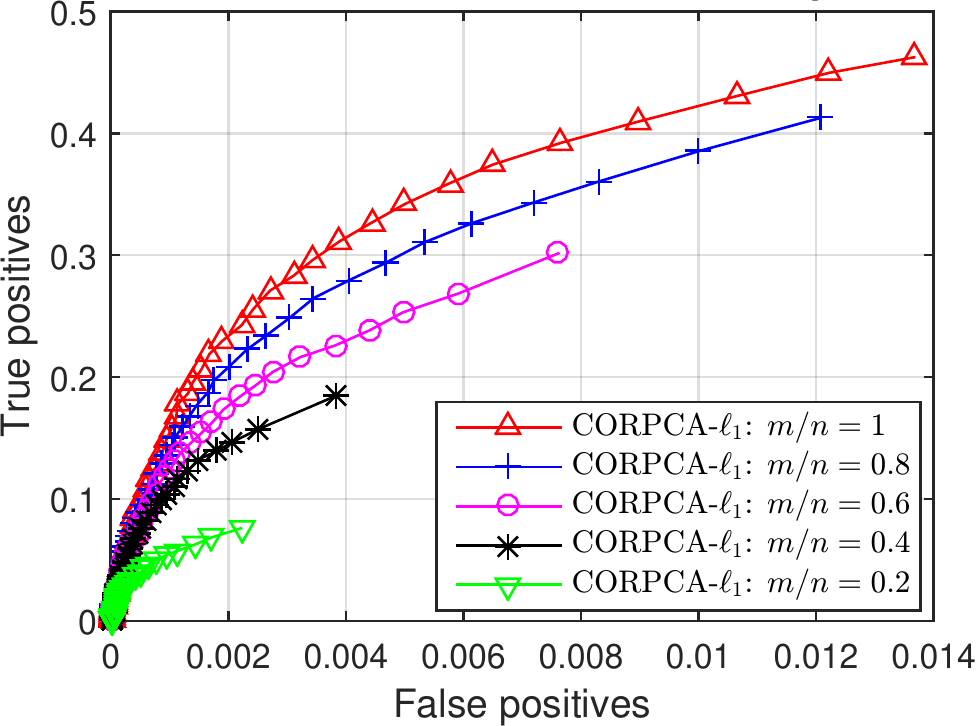}\label{bootstrapROCCorpcal1}}	
	\subfigure[\vspace{-0.15pt}\texttt{Curtain}: CORPCA-$n$-$\ell_1$]{
		\includegraphics[width=0.315\textwidth]{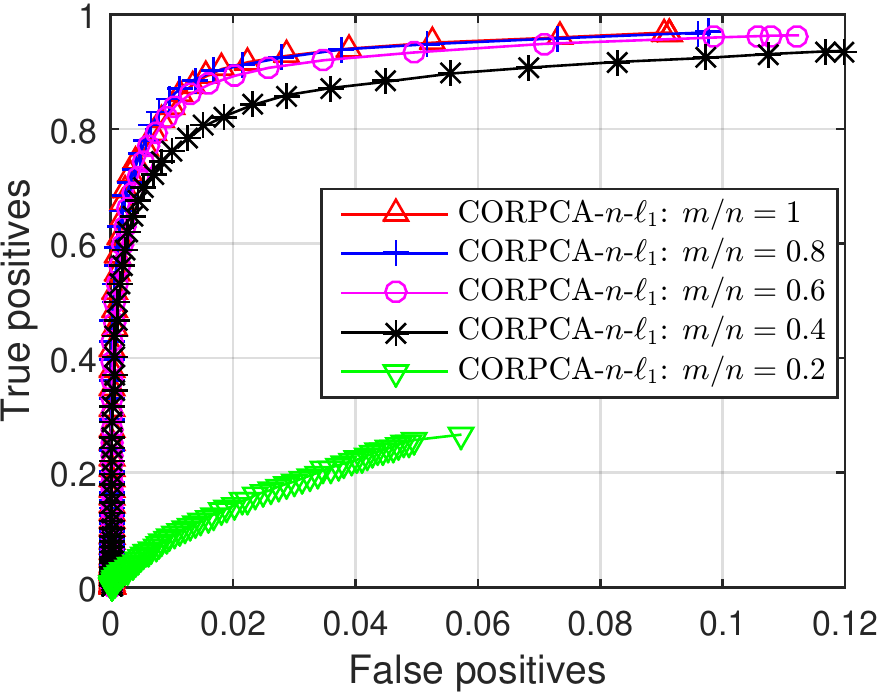}\label{curtainROCCorpcanl1}}
	\subfigure[\vspace{-0.15pt}\texttt{Curtain}: CORPCA-$\ell_1$-$\ell_1$]{
		\includegraphics[width=0.325\textwidth]{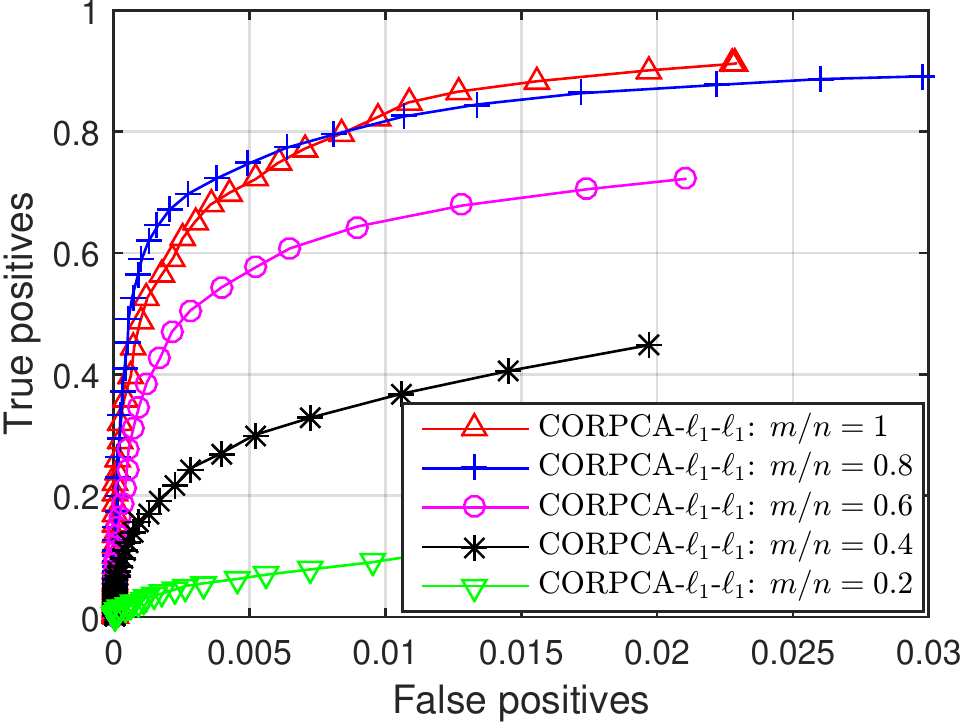}\label{curtainROCCorpcal1l1}}	
	\subfigure[\vspace{-0.15pt}\texttt{Curtain}: CORPCA-$\ell_1$]{
		\includegraphics[width=0.32\textwidth]{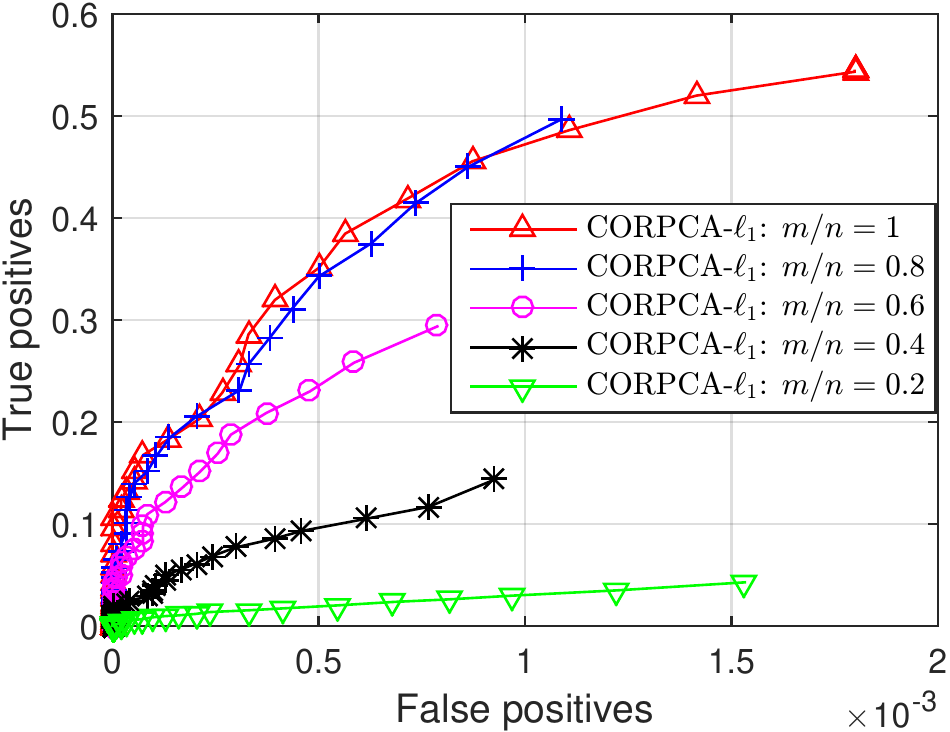}\label{curtainROCCorpcal1}}	
	\caption{ROC for CORPCA deploying different sparse recovery algorithms with different measurement rates for \texttt{Bootstrap}.}\label{bootstrapCurtainROCCorpca}
	\vspace{-0.14pt}
\end{figure*}
\begin{figure*}[tp!]
	\centering
	\subfigure[\vspace{-0.15pt}\texttt{Bootstrap}]{
		\includegraphics[width=0.33\textwidth]{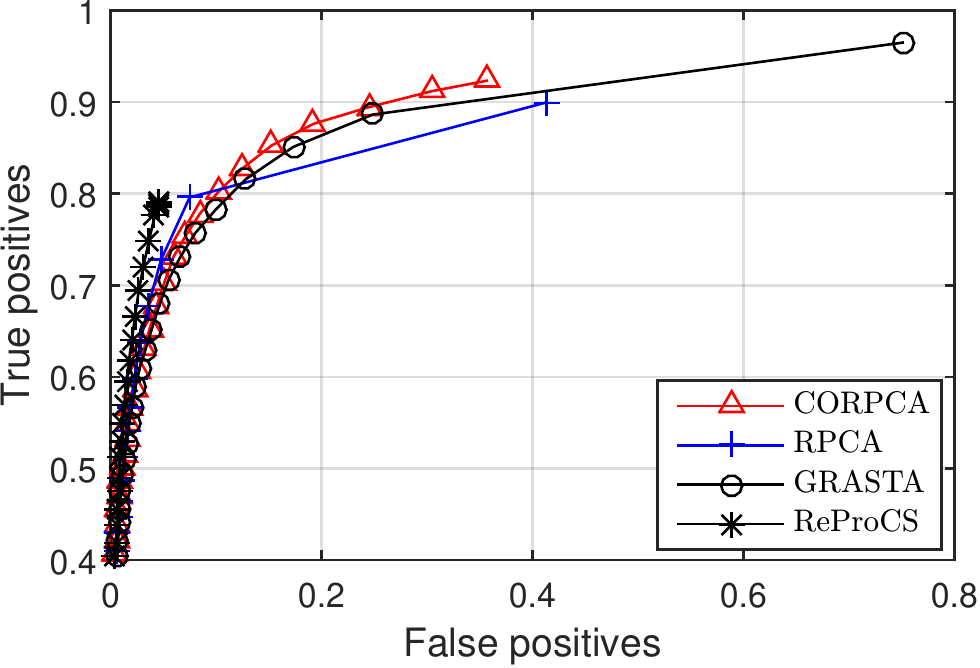}\label{bootstrapROCFull}}
	\subfigure[\vspace{-0.15pt}\texttt{Curtain}]{
		\includegraphics[width=0.33\textwidth]{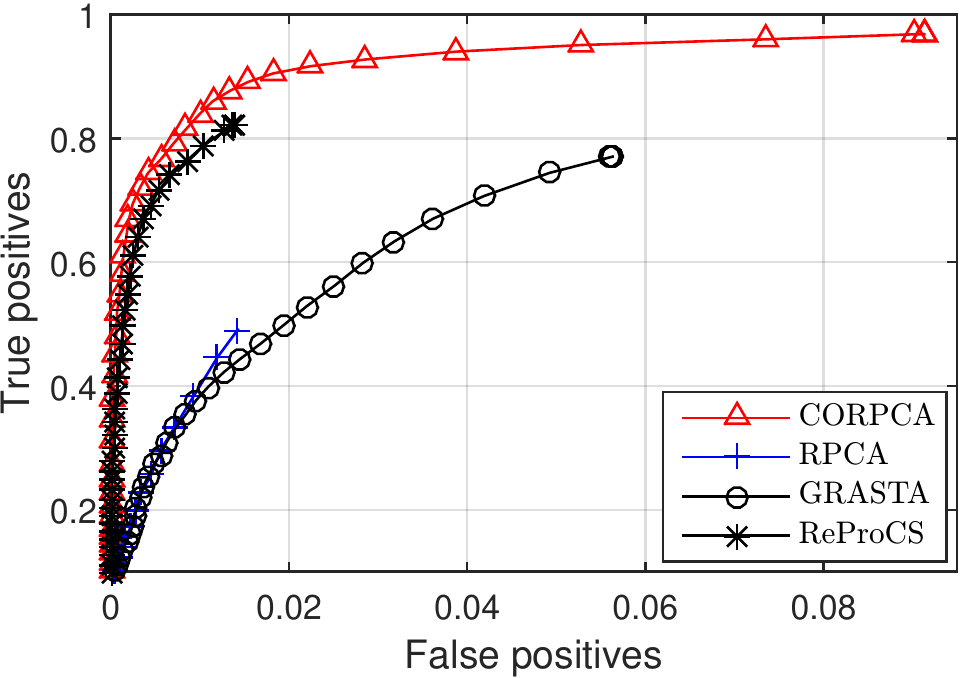}\label{curtainROCFull}}
	\vspace{-4pt}
	\caption{ROC for the different separation methods with full data
		.}\label{ROCFull}
	\vspace{-0pt}
\end{figure*}
\vspace{-0pt}
\begin{figure*}[tp!]
	\centering
	\hspace{-5pt}\subfigure[\vspace{-0.15pt}\texttt{Bootstrap}]{
		\includegraphics[width=0.345\textwidth]{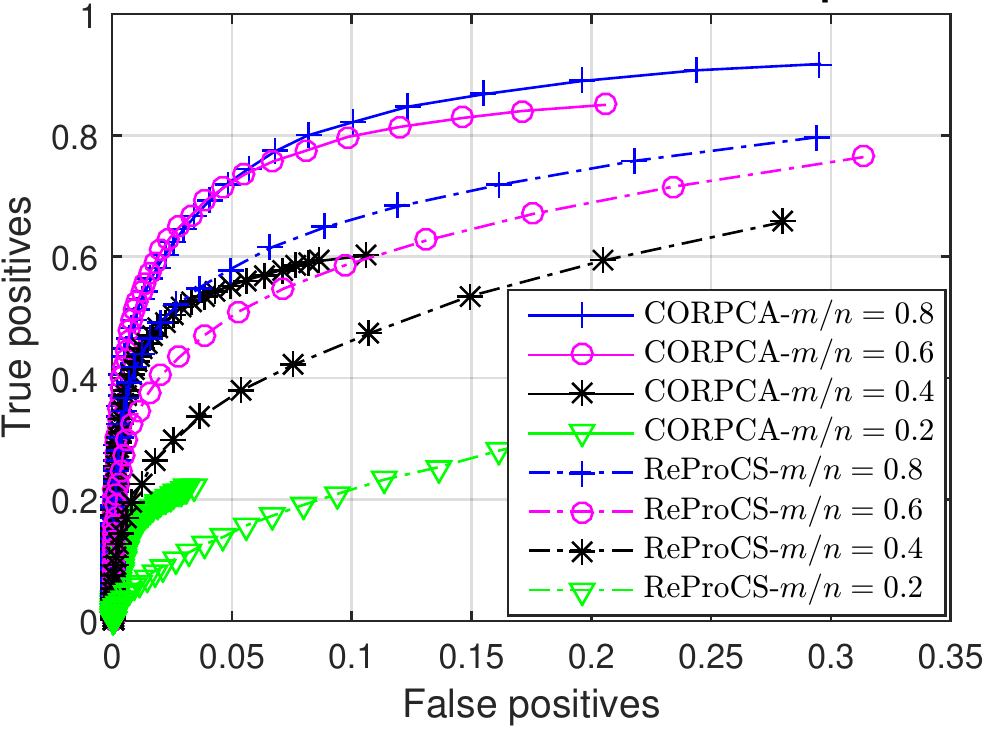}\hspace{-5pt}\label{bootstrapROCReprocs}}
	\subfigure[\vspace{-0.15pt}\texttt{Curtain}]{
		\includegraphics[width=0.345\textwidth]{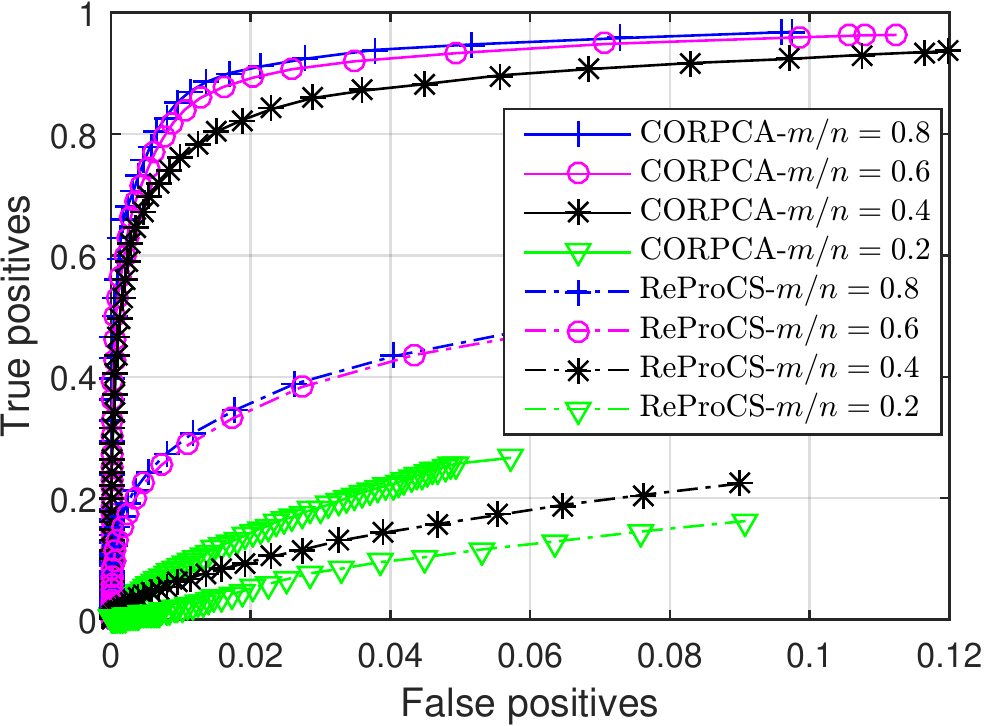}\hspace{-5pt}\label{curtainROCReprocs}}
	\vspace{-4pt}
	\caption{ROC for CORPCA and ReProCS with compressive measurements.}\label{ROCReprocs}
	\vspace{2pt}
\end{figure*}
We evaluate quantitatively the separation performance via the \textit{receiver operating curve} (ROC) metric \cite{MDikmen}. The metrics \textit{True positives} and \textit{False positives} are defined as in \cite{MDikmen}. 

The ROC results for CORPCA deploying different recovery algorithms, CORPCA-$n$-$\ell_1$, CORPCA-$\ell_1$-$\ell_1$, and CORPCA-$\ell_1$, are depicted in Fig. \ref{bootstrapCurtainROCCorpca} for \texttt{Bootstrap} and \texttt{Curtain} sequences. It is clear that CORPCA-$\ell_1$-$\ell_{1}$ [see Figs. \ref{bootstrapROCCorpcal1l1}, \ref{curtainROCCorpcal1l1}] and CORPCA-$\ell_{1}$ [see Figs. \ref{bootstrapROCCorpcal1}, \ref{curtainROCCorpcal1}] give worse curves compred with CORPCA-$n$-$\ell_1$ [see Figs. \ref{bootstrapROCCorpcanl1}, \ref{curtainROCCorpcanl1}]. Specifically, CORPCA-$n$-$\ell_1$ gives high performances with a small number of measurements, e.g., for \texttt{Bootstrap} until $m/n=0.6$ [see Fig. \ref{bootstrapROCCorpcanl1}] and for \texttt{Curtain} until $m/n=0.4$ [see Fig. \ref{curtainROCCorpcanl1}]. 

Fig. \ref{ROCFull} illustrates the ROC results when assuming full data access, i.e., $m/n=1$, of CORPCA, RPCA, GRASTA, and ReProCS. The results show that CORPCA delivers higher performance than the other methods, especially for the \texttt{Curtain} video sequence [c.f., Fig. \ref{curtainROCFull}]. Furthermore, we compare the foreground recovery performance of CORPCA against ReProCS for different compressive measurement rates: $m/n=\{0.8;0.6;0.4;0.2\}$. The ROC results in Fig. \ref{ROCReprocs} show that CORPCA achieves a relatively high performance with a small number of measurements. 

\section{Conclusion}
\label{conclusion}
This paper proposed a compressive online robust PCA with prior information for recursive separation. The proposed CORPCA method can process one data frame per time instance from compressive measurements. CORPCA efficiently incorporates multiple prior information based on the $n$-$\ell_{1}$ minimization problem. To express the performance, we established the theoretical bounds on the number of required measurements to guarantee successful recovery. The numerical results showed the efficiency of CORPCA in terms of both the practical results and the theoretical bounds. The results also revealed the advantage of incorporating prior information by employing $n$-$\ell_{1}$ minimization to achieve significant improvements over that of using the existing $\ell_{1}$-$\ell_{1}$ and $\ell_{1}$ minimization methods. We also test our method on the compressive video separation application using video data. The results revealed the advantage of incorporating prior information by employing $n$-$\ell_{1}$ minimization and demonstrated the superior performance improvement offered by CORPCA compared to the existing methods. 
\section{Appendix}\label{appendix}
\subsection{Prior Art in Compressed Sensing}\label{fundamentalRecovery}
The problem of recursive signal recovery employs Compressed Sensing (CS) \cite{CandesTIT06,DonohoCOM06} to recover sparse components from low-rank components \cite{VaswaniZ16,MotaTSP16}, which is considered as a dynamic CS. Therefore, we review some related CS results before stating our contributions. Let $\bx\in\mathbb{R}^{n}$ denote a high-dimensional sparse vector. The source $\bx$ can be reduced by sampling via a linear projection at the encoder~\cite{CandesTIT06}. We denote a random measurement matrix by $\mathbf{\Phi}\in \mathbb{R}^{ m\times n}$ (with~$m\ll{n}$), whose elements are sampled from an i.i.d. Gaussian distribution. Thus, we get a measurement vector $\by=\mathbf{\Phi}\bx$, with $\by\in\mathbb{R}^{ m}$. At the decoder, $\bx$ can be recovered by solving the Basis Pursuit problem~\cite{CandesTIT06,DonohoCOM06}:
\begin{equation}\label{l1-norm}
\min_{\bx} \|\bx\|_{1} \mathrm{~subject~to~} \by=\mathbf{\Phi}\bx,
\end{equation}
where $\|\bx\|_{p}:=(\sum_{i=1}^{n}|x_{i}|^{p})^{1/p}$ is the $\ell_{p}$-norm of $\bx$ wherein $x_{i}$ is an element of $\bx$.
Problem \eqref{l1-norm} becomes an instance of finding a general solution:
\begin{equation}\label{l1-general}
\min_{\bx}\{H(\bx) = f(\bx) + g(\bx)\},
\end{equation}
where $f:=\mathbb{R}^{n} \rightarrow\mathbb{R}$ is a smooth convex function and~$g:=\mathbb{R}^{n} \rightarrow\mathbb{R}$ is a continuous convex function, which is possibly non-smooth. Problem \eqref{l1-norm} emerges from Problem~\eqref{l1-general} when setting $g(\bx)=\lambda \|\bx\|_{1}$ and $f(\bx)=\frac{1}{2}\|\mathbf{\Phi}\bx-\by\|^{2}_{2}$, with Lipschitz constant $L_{\nabla f}$ \cite{Beck09}. 

We consider how many measurements are required for successful recovery. The classical $\ell_{1}$ minimization problem in CS \cite{DonohoTIT06,DonohoCOM06,CandesTIT06} requires $m_{\ell_{1}}$ measurements \cite{Venkat12,MotaGLOBALSIP14,MotaARXIV14} for successful reconstruction, bounded as
\begin{equation}\label{l1 bound}
m_{\ell_{1}} \geq 2s_{0}\log\frac{n}{s_{0}} + \frac{7}{5}s_{0} + 1,
\end{equation}
where $s_{0}:=\|\bx\|_0=|\{i: x_{i}\neq 0\}|$ denotes the number of nonzero elements in $\bx$ as the support of $\bx$, with $|.|$ denoting the cardinality of a set and $\|\cdot\|_0$ being the $\ell_0$-pseudo-norm.

Furthermore, the $\ell_{1}$-$\ell_{1}$ minimization approach~\cite{MotaGLOBALSIP14,MotaARXIV14,MotaICASSP15} reconstructs $\bx$ given a signal $\bz \in \mathbb{R}^{n}$ as side information by solving the following problem:
\begin{equation}\label{l1-l1minimization}
\min_{\bx}\Big\{\frac{1}{2}\|\mathbf{\Phi}\bx-\by\|^{2}_{2} + \lambda (\|\bx\|_{1}+\|\bx-\bz\|_{1})\Big\}.
\end{equation}
The bound on the number of measurements required by Problem~\eqref{l1-l1minimization} to successfully reconstruct $\bx$ depends on the quality of the side information signal $\bz$ as \cite{MotaGLOBALSIP14,MotaARXIV14,MotaICASSP15}
\begin{equation}\label{l1-l1 bound}
m_{\ell_{1}\text{-}\ell_{1}} \geq 2\overline{h}\log\Big(\frac{n}{s_{0}+\xi/2}\Big) + \frac{7}{5}\Big(s_{0}+\frac{\xi}{2}\Big) + 1,
\end{equation}
where
\begin{subequations}\label{l1-l1 sparse set}
	\begin{align}
	\xi:&=|\{i:z_{i}\neq x_{i}=0\}|-|\{i:z_{i}= x_{i}\neq 0\}|\nonumber \\
	\bar{h}:&=|\{i:x_{i}>0,x_{i}>z_{i}\}\cup\{i:x_{i}<0,x_{i}<z_{i}\}|\label{l1-l1 sparse setHBar},\nonumber
	\end{align}
\end{subequations}
wherein $x_{i}$, $z_{i}$ are corresponding elements of $\bx$, $\bz$. It has been shown that Problem \eqref{l1-l1minimization} improves over Problem \eqref{l1-norm} provided that the side information has good enough quality~\cite{MotaGLOBALSIP14,MotaARXIV14}. 
\subsection{Background on Measurement Condition}\label{backgroundConvex} 
We restate some key definitions and conditions in convex optimization \cite{Venkat12,Tropp14}, which are used in the derivation of the measurement bounds. 

A \textit{convex cone} $C\subset \mathbb{R}^{n}$ is a convex set that satisfies $C=\tau C$, $\forall\tau\geq0$~\cite{Tropp14}. 
For the cone $C\subset \mathbb{R}^{n}$, a \textit{polar cone} $C^{\circ}$ is the set of outward normals of $C$, defined by
\begin{equation}\label{polarCone}
C^{\circ}:=\{\bu\in \mathbb{R}^{n}:\bu^{\mathrm{T}}\bx\leq 0, ~\forall\bx\in C\}.
\end{equation}
A \textit{descent cone}~\cite[Definition 2.7]{Tropp14} $\mathcal{D}(g,\bx)$, alias \textit{tangent cone}~\cite{Venkat12}, of a convex function $g:=\mathbb{R}^{n} \rightarrow\mathbb{R}$ at a point $\bx\in \mathbb{R}^{n}$---at which $g$ is not increasing---is defined as
\begin{equation}\label{descentCone}
\mathcal{D}(g,\bx):=\bigcup\limits_{\tau \geq 0}\{\by\in \mathbb{R}^{n}:g(\bx+\tau \by)\leq g(\bx)\},
\end{equation}
where $\bigcup$ denotes the union operator.

Recently, a new summary parameter called the \textit{statistical dimension} $\delta(C)$ of cone $C$ \cite{Tropp14} is introduced to estimate the convex cone~\cite[Theorem
4.3]{Tropp14}. The statistical dimension can be expressed in terms of the polar cone $C^{\circ}$ as~\cite[Proposition 3.1]{Tropp14}
\begin{equation}\label{statDimComputePolar}
\delta(C):=\mathbb{E}_{\bg}\big[\mathrm{dist}^{2}(\textbf{\textsl{g}},C^{\circ})].
\end{equation}
\subsubsection{Gaussian Width}
The Gaussian width \cite{Venkat12} is a summary parameter for convex cones; it is used to measure the aperture of a convex cone. For a convex cone $C\subset \mathbb{R}^{n}$, considering a subset $C\cap \mathbb{S}^{n-1}$ where $\mathbb{S}^{n-1}\subset \mathbb{R}^{n}$ is a unit sphere, the \textit{Gaussian width} \cite[Definition 3.1]{Venkat12} is defined as
\begin{equation}\label{gaussianWidth}
\omega(C):=\mathbb{E}_{\bg}[\sup_{\bu \in C\cap \mathbb{S}^{n-1}} \textbf{\textsl{g}}^{\mathrm{T}}\bu].
\end{equation}
where $\bg\sim\mathcal{N}(0,\mathbf I_{n})$ is a vector of $n$ independent, zero-mean, and unit-variance Gaussian random variables and $\mathbb{E}_{\bg}[\cdot]$ denotes the expectation with respect to $\bg$.
The Gaussian width~\cite[Proposition 3.6]{Venkat12} can further be bounded as
\begin{equation}\label{gaussianWidthBound}
\omega(C)\leq \mathbb{E}_{\bg}\big[\mathrm{dist}(\textbf{\textsl{g}},C^{\circ})],
\end{equation}
where $\mathrm{dist}(\textbf{\textsl{g}},C^{\circ})$ denotes the Euclidean distance of $\textbf{\textsl{g}}$ with respect to the set $C^{\circ}$, which is in turn defined as
\begin{equation}\label{euclideanDistance}
\mathrm{dist}(\textbf{\textsl{g}},C^{\circ}):=\min_{\bu}\{\|\textbf{\textsl{g}}-\bu\|_{2}:\bu\in C^{\circ}\}.
\end{equation}


\subsubsection{Measurement Condition}
An optimality condition~\cite[Proposition 2.1]{Venkat12}, \cite[Fact 2.8]{Tropp14} for linear inverse problems states that $\bx_{0}$ is the unique solution of \eqref{l1-general} if and only if
\begin{equation}\label{optimalCondition}
\mathcal{D}(g,\bx_{0})\cap \mathrm{null}(\mathbf{\Phi})=\{\mathbf{0}\},
\end{equation}
where $\mathrm{null}(\mathbf{\Phi})\hspace{-0.2pt}:=\hspace{-0.2pt}\{\bx\in\mathbb{R}^{n}\hspace{-0.2pt}:\hspace{-0.2pt}\mathbf{\Phi}\bx\hspace{-0.2pt}=\hspace{-0.2pt}\mathbf{0}\}$ is the null space of $\mathbf{\Phi}$. We consider the number of measurements $m$ required to successfully reconstruct a given signal $\bx_{0}\in \mathbb{R}^{n}$. Corollary 3.3 in \cite{Venkat12} states that, given a measurement vector $\by=\mathbf{\Phi}\bx_{0}$, $\bx_{0}$ is the unique solution of \eqref{l1-general} with probability at least $1\hspace{-0.2pt}-\hspace{-0.2pt}\exp(-\frac{1}{2}(\sqrt{m}\hspace{-0.2pt}-\hspace{-0.2pt}\omega(\mathcal{D}(g,\bx_{0})))^{2})$ provided that $m\geq \omega^{2}(\mathcal{D}(g,\bx_{0}))+1$. 
\subsubsection{Bound on the Measurement Condition}\label{propositionAppendix}
Let us consider that the subdifferential $\partial g$ \cite{JHiriat} of a convex function $g$ at a point $\bx \in \mathbb{R}^{n}$ is given by $\partial g\hspace{-0.2pt}:=\hspace{-0.2pt}\{\bu \hspace{-0.2pt}\in\hspace{-0.2pt}\mathbb{R}^{n}\hspace{-0.2pt}: g(\by)\hspace{-0.2pt}\geq\hspace{-0.2pt} g(\bx)\hspace{-0.2pt}+\hspace{-0.2pt}\bu^{\mathrm{T}}(\by\hspace{-0.2pt}-\hspace{-0.2pt}\bx)~\text{for all}~ \by \hspace{-0.2pt}\in \hspace{-0.2pt}\mathbb{R}^{n}$\}. From \eqref{statDimComputePolar}, \eqref{gaussianWidthBound}, and~\cite[Proposition 4.1]{Tropp14}, we obtain an upper bound on $\delta(\mathcal{D}(g,\bx))$ as
\begin{equation}\label{upperBound}
\begin{split}
\omega^{2}(\mathcal{D}(g,\bx_{0}))\leq\delta(\mathcal{D}(g,\bx))\leq \min_{\tau\geq 0}\mathbb{E}_{\bg}\big[\mathrm{dist}^{2}(\textbf{\textsl{g}},\tau \hspace{-0.2pt}\cdot \hspace{-0.2pt} \partial g(\bx))\big].
\end{split}
\end{equation}

%
%
We conclude the following proposition for determine the measurement bound.
\begin{proposition}\label{propUpper}\cite[Corrollary 3.3, Proposition 3.6]{Venkat12}
	We have the successful recovery $\bx\hspace{-0.2pt}\in\hspace{-0.2pt}\mathbb{R}^{n}$ in \eqref{l1-general} that we have observed measurements $\by=\mathbf{\Phi}\bx$, where $\mathbf{\Phi}\hspace{-0.2pt}\in\hspace{-0.2pt}\mathbb{R}^{m\times n}$, provided the condition $m\geq U_{g}\hspace{-0.1pt}+\hspace{-0.1pt}1$. For the noisy measurements $\|\by\hspace{-0.2pt}-\hspace{-0.2pt}\mathbf{\Phi}\bx\|_{2}\hspace{-0.2pt}<\hspace{-0.2pt}\sigma$, we have that $\|\bx-\hat{\bx}\|_{2}<\frac{2\sigma}{1-\sqrt{\rho}}$ provided the condition $m\geq \frac{U_{g}\hspace{-0.1pt}+3/2}{\rho}$, where $\hat{\bx}$ is any solution in \eqref{l1-general} and $0<\rho<1$. The quantity $U_{g}$ is calculated given a convex norm function $g\hspace{-0.1pt}:=\hspace{-0.1pt}\mathbb{R}^{n} \hspace{-0.2pt}\rightarrow\hspace{-0.2pt}\mathbb{R}$
	\begin{equation}\label{upperBoundCompute}
	U_{g}= \min_{\tau\geq 0}\mathbb{E}_{\bg}\Big[\mathrm{dist}^{2}(\textbf{\textsl{g}}, \tau \cdot\partial g(\bx))\Big],
	\end{equation}
where $\partial g(\bx)$ is the subdifferential \cite{JHiriat} of $g(\bx)$ at a point $\bx$.
\end{proposition}

\subsection{Some supported results}\label{support}
Recall that the probability density of the normal distribution $\mathcal{N}(0,1)$
with zero-mean and unit variance $\psi(x)$ is given by
\begin{equation}\label{densityDistribution}
\psi(x):=(1/\sqrt{2\pi})e^{-x^{2}/2}.
\end{equation}

We also consider the following inequality \cite{MotaARXIV14}:
\begin{equation}
\label{logInequality}
\frac{(1-x^{-1})}{\sqrt{\pi\log(x)}}\leq\frac{1}{\sqrt{2\pi}}\leq\frac{2}{5},
\end{equation}
for all $x>1$. Moreover, adhering to the formulation in~\cite{MotaARXIV14},
we use the following inequality with $x > 0$ in our derivations:
\begin{equation}
	\label{positiveA}
	\frac{1}{\sqrt{2\pi}}\int_{x}^{\infty}(v-x)^{2}e^{-v^{2}/2}\mathrm{d}v\leq\dfrac{\psi(x)}{x}.
\end{equation}

\bibliographystyle{IEEEtran}
\bibliography{./IEEEfull,./IEEEabrv,./bibliography}

\end{document}